\documentclass[11pt]{article}

\usepackage{latexsym}
\usepackage{color,graphicx}
\usepackage{amssymb}
\usepackage{amsmath}
\usepackage{amsthm}
\usepackage{txfonts}
\usepackage{enumitem}
\usepackage{fullpage}

\newtheorem{theorem}{Theorem}[section]
\newtheorem{lemma}[theorem]{Lemma}
\newtheorem{corollary}[theorem]{Corollary}
\newtheorem{observation}[theorem]{Observation}

\newcommand{\arcCost}[1]{\texttt{cost}(#1)}
\newcommand{\arcCap}[1]{\texttt{cap}(#1)}

\newcommand{\cS}{\mathcal{S}}
\newcommand{\cP}{\mathcal{P}}

\newcommand{\cM}{\mathcal{M}}

\newcommand{\card}[1]{\left|#1\right|}
\newcommand{\payoff}[1]{\textup{\texttt{payoff}}(#1)}
\newcommand{\overlap}[2]{\textup{\texttt{ovlp}}_{#1}(#2)}

\newcommand{\Mshared}[1][]{\cM \if!#1!\else (#1) \fi}

\newcommand{\Mpriv}{\cP}

\newcommand{\complTime}[3]{C_{#1}^{#3}(#2)}  %{C_{#2}(#1,#3)}            % #1=schedule, #2=job, #3=processor
\newcommand{\tct}[1]{\varSigma(#1)}                  % #1=schedule

\newcommand{\reals}{\mathbb{R}}

\newcommand{\jobs}{\mathcal{J}}
\newcommand{\Sopt}{\cS_{\textup{opt}}}
\newcommand{\nextEps}[1]{\textup{\texttt{next}}(#1)}

\newcommand{\w}[1]{w_{#1}}
\newcommand{\p}[1]{p_{#1}}
\newcommand{\cm}[1]{c_{#1}}

\newcommand{\procStart}{\medskip\hrule\vspace*{3pt}\noindent}
\newcommand{\procEnd}{\vspace*{1pt}\hrule\medskip}
\newcommand{\procSequential}{\textup{\texttt{MakeSequential}}}
\newcommand{\procTransfer}{\textup{\texttt{Transfer}}}

\usepackage[dvipsnames]{xcolor}

\begin{document}

\title{\textbf{Shared Processor Scheduling of Multiprocessor Jobs}}

\author{
  Dariusz Dereniowski\footnote{Corresponding author. Email: deren@eti.pg.edu.pl}\\
  \small{\emph{Faculty of Electronics,}}\\
  \small{\emph{Telecommunications and Informatics},}\\
  \small{\emph{Gda{\'n}sk University of Technology},}\\
  \small{\emph{Gda{\'n}sk, Poland}}
\and
  Wies{\l}aw Kubiak\\
  \small{\emph{Faculty of Business Administration},}\\
  \small{\emph{Memorial University},}\\
  \small{\emph{St. John's, Canada}}
}

%\date{}

\maketitle

\begin{abstract} We study shared processor scheduling of \emph{multiprocessor} weighted jobs where each job can be executed on its private processor and simultaneously on possibly \emph{many} processors shared by all jobs in order to reduce their completion times due to processing time overlap. Each of $m$ shared processors may charge different fee but otherwise the processors are identical. The total weighted overlap of all jobs is to be maximized. This problem is key to subcontractor scheduling in extended enterprises and supply chains, and divisible load scheduling in computing. We prove that, quite surprisingly, \emph{synchronized} schedules that complete each job using shared processors at the same time on its private and shared processors include optimal schedules. We show that optimal \emph{$\alpha$-private} schedules that require each job to use its private processor for at least $\alpha=1/2+1/(4(m+1))$ of the time required by the job guarantee more than an $\alpha$ fraction of the total weighted overlap of the optimal schedules. This gives an $\alpha$-approximation algorithm that runs in strongly polynomial time for the problem, and improves the $1/2$-approximation reported recently in the literature to $5/8$-approximation for a single shared processor problem. The computational complexity of the problem, both single and multi-shared processor, remains open. We show however an LP-based optimal algorithm for \emph{antithetical} instances where for any pair of jobs $j$ and $i$, if the processing time of $j$ is smaller than or equal to the processing time of $i$, then the weight of $j$ is greater than or equal to the weight of $i$.
\end{abstract}


\section{Introduction}

Quick-response industries are characterized by volatile demand and inflexible capacities. The agents (companies) in such industries  need to supplement their \emph{private} capacity by adapting their 
 extended enterprises and supply chains to include 
subcontractors with their own capacity.  This capacity of subcontractors however is often \emph{shared} between other independent supply chains which can cause undesirable and difficult to control bottlenecks in those supply chains. A well-documented real-life example of this issue has been reported in Boeing's Dreamliner supply chain where the overloaded schedules of subcontractors, each working with multiple suppliers, resulted in long delays in the overall production due dates, see Vairaktarakis~\cite{V13}. 

The use of subcontractor's shared processor (capacity) benefits an agent only if it can reduce the agent's job (order) completion time at a competitive enough cost. Hence, the subcontractor's shared processor should never be used, and paid for by the agent, as long as the agent's private processor remains available. We reasonably assume that the cost of using subcontractor's processor is higher than this of the agent's private processor. Moreover the agent's private processor should never remain idle as long as the agent's job remains unfinished. Therefore, only a simultaneous execution, or \emph{overlap}, on both private and shared processors reduces completion time. The total  (weighted) overlap is the objective function studied
in this paper. This objective function is closely related to the total completion time objective traditionally used in scheduling. The total completion time can be reduced by an increase of the total overlap resulting from the simultaneous execution of jobs on private and shared processors. However,  we need to emphasize that the two objectives exist for different practical reasons. The minimization of total completion time minimizes mean flow time and thus by Little's Law minimizes average inventory in the system. The maximization of the total overlap on the other hand maximizes the total net payoff resulting from completing jobs earlier  thanks to the use of shared processors (subcontractors). This different focus sets the total overlap objective apart from the total completion time objective, and makes it a key objective in scheduling shared processors, Dereniowski and Kubiak~\cite{DK16}. 

The reduction of completion time of a job due to the overlap depends on whether only a single shared processor or multiple shared processors can be used simultaneously by the job. For instance, an order of size 12 can be completed in 6 units of time (assuming it takes one unit of time to complete the order of size one) at the earliest if only a single shared processor is allowed to process the order simultaneously with private processor. The order is then split in half between the private and the shared processor both working simultaneously on the job in the time interval $(0,6)$. The resulting overlap equals $6$, and the job is not executed by any other shared processor in the interval.  This constraint has been imposed in the literature thus far, see Vairaktarakis and Aydinliyim~\cite{VairaktarakisAydinliyim07},  Hezarkhani and Kubiak~\cite{HK15}, and  Dereniowski and Kubiak~\cite{DK17},  \cite{DK16}.  We refer to the constraint as a \emph{single processor} (\emph{SP}) job mode.
This paper relaxes the constraint and permits a job to be processed simultaneously on its private and possibly \emph{more} than one shared processor. For instance, the job of size 12 can be completed in 4 units of time by executing it in the interval $(0,4)$ on its private processor and simultaneously in the intervals $(0,4)$, $(0,3)$ and $(1,2)$ on three different shared processors. The resulting total overlap equals $8$. We refer to this relaxation as a \emph{multiprocessor} (\emph{MP}) job mode.  To our knowledge this mode of execution of jobs has been first studied by Blazewicz, Drabowski and Weglarz~\cite{BDW86} and refereed to as multiprocessor jobs in the literature. The multiprocessor jobs gained prominence in distributed computing where the processing by the nodes of a shared network of processors as well as possible communications between the nodes overlap in time so that the completion time (makespan) for the whole job (referred to as divisible load) is shorter than the processing of the whole load by a single node, Bharadwaj, Ghose, and Robertazzi~\cite{HBGR03}.
Bharadwaj, Ghose, and Robertazzi~\cite{HBGR03} and Drozdowski~\cite{D09} survey many real-life applications that satisfy the divisibility property. 

The shared processors may also charge different fees, $c_i$, and the jobs may have different weights $w_j$.
Then, the contribution to the total payoff of a job piece of length $l$ is $l$ times the difference between its weight and the shared processors' fee.
Thus, the execution in the interval $(0,4)$ on one shared processor may cost the same as the execution in the interval $(1,2)$ on another shared processor if the latter is four times more expensive than the former. The shared processors with different costs will studied in this paper.
Figure~\ref{fig:example2} illustrates the difference between the two modes, $SP$ and $MP$.

\begin{figure}[ht!]
\begin{center}
 \includegraphics[scale=1.1]{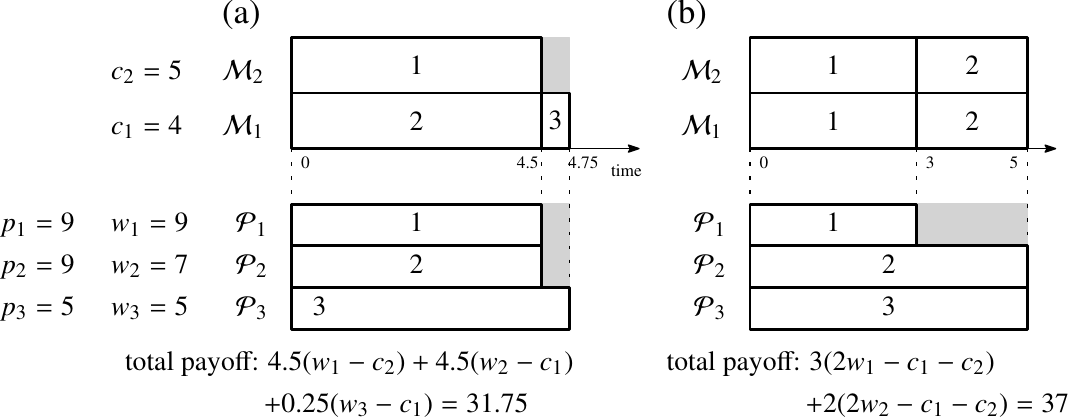}
\end{center}
\caption{The example illustrates that allowing jobs to be executed in \emph{MP} mode simultaneously on several shared processors ($\cM_1$ and $\cM_2$) may be beneficial for some problem instances: (a) an optimal (synchronized) schedule for the \emph{SP} mode; (b) an optimal (synchronized) schedule for the \emph{MP} mode.
In this input instance, the shared processors' fees are $c_1=4$ and $c_2=5$, the jobs' processing times are $p_1=p_2=9$, $p_3=5$ and their weights are $w_1=9$, $w_2=7$, $w_3=5$.}
\label{fig:example2}
\end{figure}

It is quite remarkable that regardless of the mode of job execution on shared processors, and job weights there always exist optimal schedules that are \emph{synchronized}, i.e., each agent using shared processors has its job completed on private and shared processors at the same time (for a formal definition of synchronized schedules for the \emph{MP} mode see Section~\ref{sec:outline}). This has been shown for the \emph{SP} job mode by Vairaktarakis and Aydinliyim  \cite{VairaktarakisAydinliyim07},  Hezarkhani and Kubiak~\cite{HK15}, and  Dereniowski and Kubiak~\cite{DK16}, and \cite{DK17}. In this paper we show it for the \emph{MP} job mode. We return to the paper outline later in the introduction in Section~\ref{sec:outline} to give more details.

\subsection{Related Work and Applications}

The shared processor scheduling problem with a \emph{single} shared processor has been studied by  Vairaktarakis and Aydinliyim  \cite{VairaktarakisAydinliyim07},  Hezarkhani and Kubiak \cite{HK15}, and  Dereniowski and Kubiak 
%\cite{DK16} 
\cite{DK17}.
Vairaktarakis and Aydinliyim  \cite{VairaktarakisAydinliyim07} consider the \emph{unweighted} problem with each job allowed to use at most  one time interval on the shared  processor. This case is sometimes referred to as \emph{non-preemptive} since jobs are not allowed preemption on the shared processor.  \cite{VairaktarakisAydinliyim07} proves that there are optimal schedules that complete job execution on its private and the shared processor at the same time, we call such schedules \emph{synchronized}.
It further shows that this guarantees that sequencing jobs in non-decreasing order of their processing times leads to an optimal solution for the case. We refer to such schedules as \emph{processing time ordered}, see    \cite{DK17}. Interestingly, the processing time ordered schedules guarantee that each job uses exactly one nonempty time interval on the shared processor. \cite{HK15} observes that the processing time ordered schedules also give optimal solutions to the \emph{preemptive} unweighted problem, 
where more than one interval can be used by a job on the shared processor.  \cite{DK17}  considers the \emph{weighted} problem. It observes that for the weighted problem it no longer holds that each job occupies a non-empty interval on the shared processor in optimal schedules, there may exist jobs processed on their private processors only in each optimal schedule. It shows that there always exist optimal schedules that are synchronized, gives a $\frac{1}{2}$- approximation algorithm for the problem, and shows that the $\frac{1}{2}$ bound for the algorithm is tight. It also extends earlier result for the unweighted problem by proving  that the processing time ordered schedules are optimal for antithetical instances, i.e., the ones for which there exists ordering of jobs that is simultaneously non-decreasing with respect to processing times and non-increasing with respect to the weights. The complexity status of the weighted problem with a single shared processor remains open.

Vairaktarakis and Aydinliyim \cite{VairaktarakisAydinliyim07}, Vairaktarakis  \cite{V13}, and Hezarkhani and Kubiak \cite{HK15} focus on the \emph{tension} between the agents and the subcontractor in the decentralized system where each agent strives to complete its job as early as possible and needs to compete with other agents for the shared processor, and the subcontractor who strives to have the shared processor occupied as long as possible to maximize its payoff. The tension calls for coordinating mechanisms to ensure the efficiency. Hezarkhani and Kubiak \cite{HK15} show such coordination mechanism for the unweighted problem, and give examples to prove that such mechanisms do not exist for the problem with weighted jobs.

Dereniowski and Kubiak \cite{DK16} consider shared \emph{multi-processor} problem. They however, contrary to this paper, assume that each job can only be processed by its private processor and at most \emph{one} out of many shared processors, the \emph{SP} mode. Besides no distinction is made between the shared processors, in particular the costs of using  shared processors are the same for all of them. \cite{DK16} proves that synchronized optimal schedules always exist for weighted multi-processor instances. However, the wighted problem is NP-hard in the strong sense. For the multi-processor problem with equal weights for all jobs, \cite{DK16} gives an efficient, polynomial-time algorithm running in time $O(n \log n)$.

The motivation to study the shared processor scheduling problem comes from diverse applications. Vairaktarakis and Aydinliyim  \cite{VairaktarakisAydinliyim07}, \cite{TAGV} consider it in the context of  supply chains and extended enterprises where subcontracting allows jobs to reduce their completion times by using a shared subcontractor's processor. Bharadwaj et. al.  \cite{HBGR03}, and Drozdowski~\cite{D09}  use the divisible load scheduling to reduce a job completion time in parallel and distributed computer systems, and  Anderson \cite {A81} argues for using batches of potentially infinitely small items that can be processed independently of other items of the batch in scheduling job-shops. We refer the reader to Dereniowski and Kubiak \cite{DK16} for more details on these applications.

\subsection{Problem Formulation} \label{sec:problem}

We are given a set $\jobs$ of $n$ preemptive jobs with non-negative processing times $\p{j}$ and weights $\w{j}$, $j\in\jobs$.
With each job $j\in\jobs$ we associate its \emph{private} processor denoted by $\Mpriv_j$.
Moreover, $m\geq 1$ \emph{shared} processors $\Mshared_1,\ldots,\Mshared_m$ are available for all jobs; processor $\Mshared_i$ has \emph{cost} $\cm{i}$, $i\in\{1,\ldots,m\}$. Without loss of generality we always assume $c_1\leq \cdots\leq c_m$ in this paper.

A schedule $\cS$ selects for each job $j\in\jobs$:
\begin{enumerate} [label={\normalfont{(\roman*)}}]
 \item\label{it:d1} a (possibly empty) collection of open and pairwise disjoint maximal time intervals $I_{i,j}^1,\ldots,I_{i,j}^{l(i,j)}$ in which the job $j$ executes on shared processor $\Mshared_i$ for each $i\in\{1,\ldots,m\}$, any $I_{i,j}^{k}$ is called a \emph{piece} of job $j$ on processor $i$ or simply \emph{piece} of job $j$ if the processor is obvious from the context, and
 \item\label{it:d2} a \emph{single} time interval $(0, \complTime{\cS}{j}{\Mpriv})$ in which $j$ executes on its private processor $\Mpriv_j$.
\end{enumerate}
In a \emph{feasible} schedule, the total length of all these intervals (the ones in~\ref{it:d1} and the one in~\ref{it:d2}) equals $\p{j}$:
\begin{equation} \label{L1}
\p{j} = \complTime{\cS}{j}{\Mpriv} + \sum_{i=1}^{m}\sum_{k=1}^{l(i,j)}\card{I_{i,j}^{k}}
\end{equation}
for each $j\in\jobs$.
Moreover, each shared processor $i$ can execute at most one job at a time, i.e. the unions of  all pieces of different jobs $j$ and $j'$ on processor $i$ are disjoint for each shared processor $\Mshared_i$, $i\in\{1,\ldots,m\}$, or formally 
\begin{equation} \label{L2}
\left(\bigcup_{k=1}^{l(i,j)} I_{i,j}^k\right) \cap \left(\bigcup_{k'=1}^{l(i,j')} I_{i,j'}^{k'}\right) = \emptyset,
\end{equation}
for any two different jobs $j$ and $j'$ and any shared processor $\Mshared_i$, $i\in\{1,\ldots,m\}$. Without loss of generality we also assume that
\begin{equation} \label{L3}
I_{i,j}^{k} \subseteq (0, \complTime{\cS}{j}{\Mpriv})
\end{equation}
for each $i$, $j$, and $k$ in a feasible schedule.

We re-emphasize that, contrary to earlier literature on shared multi-processor scheduling \cite{DK16}, we allow the pieces of the same job \emph{not} to be disjoint (or simply to overlap) on different shared processors like in the \emph{MP} job mode introduced in \cite{BDW86} and used in \cite{HBGR03} and \cite{D09} for instance. Observe that the private processor $\Mpriv_j$ can only execute job $j$ but none of the other jobs.

\medskip
Given a feasible schedule $\cS$, for each job $j\in\jobs$ we call any pair $(\Mshared_i,I)$ an \emph{overlap of} $j$ \emph{on} $\Mshared_i$ if $I$ is a time interval of maximum length where  $j$ executes on both its private processor $\Mpriv_j$ and the shared processor $\Mshared_i$ simultaneously; we say that $\card{I}$ is the \emph{length} of the overlap $(\Mshared_i,I)$.
Note that $I\subseteq(0,\complTime{\cS}{j}{\Mpriv})$.
Then, the \emph{total overlap} $\overlap{\cS}{j,\Mshared_i}$ of  $j$ on $\Mshared_i$ equals the sum of lengths of all overlaps of $j$ on $\Mshared_i$.
The \emph{total weighted overlap} of $\cS$ equals
\begin{equation} \label{L4}
\tct{\cS}=\sum_{i=1}^{m}\sum_{j\in\jobs}\overlap{\cS}{j,\Mshared_i}(\w{j}-\cm{i}).
\end{equation}

To illustrate we give an example in Figure~\ref{fig:example}. The example also gives some intuitions as to how optimal schedules look like --- for more details see a discussion at the end of the next section.
\begin{figure}[ht!]
\begin{center}
 \includegraphics[scale=1.1]{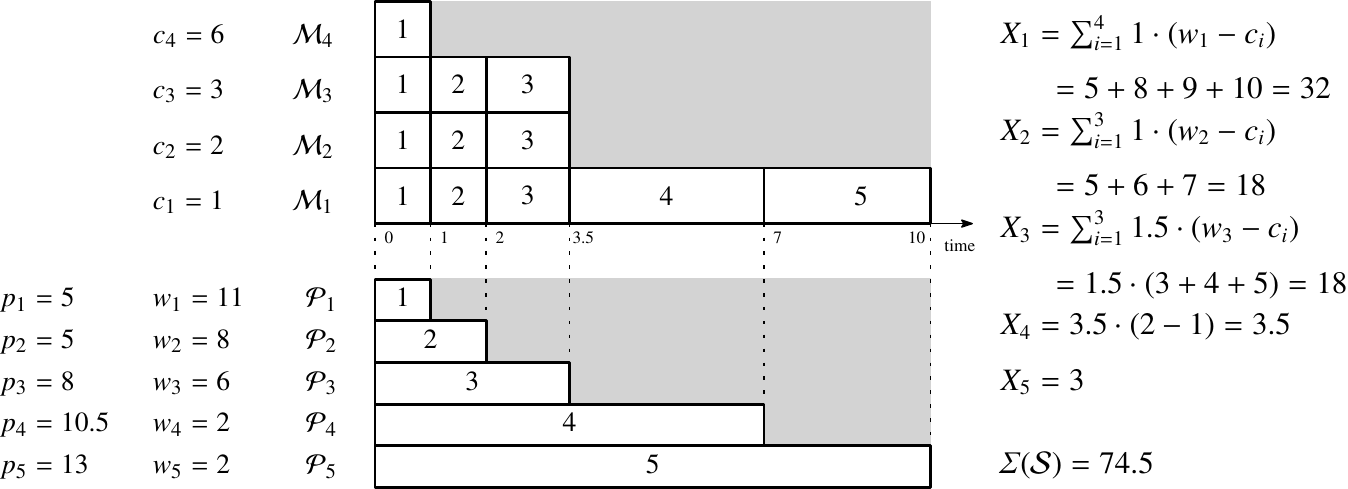}
\end{center}
\caption{An optimal schedule for an input instance with $\jobs=\{1,\ldots,5\}$ and $4$ shared processors $\Mshared_1,\ldots,\Mshared_4$. Here $X_j$ denotes the contribution of a job $j$ to the total weighted overlap of the schedule, $X_j=\sum_{i=1}^m\overlap{\cS}{j,\Mshared_i}(\w{j}-\cm{i})$}
\label{fig:example}
\end{figure}

\subsection{Outline}\label{sec:outline}

The main structural property of optimal schedules proved in this paper is schedule synchronization.
We say that a feasible schedule $\cS$ is \emph{synchronized} if there exists a subset of jobs $\{j_1,\ldots,j_k\}\subseteq\jobs$ such that:
\begin{itemize}
 \item each job $j\notin\{j_1,\ldots,j_k\}$ executes only on its private processor $\Mpriv_j$ in $(0,p_j)$ in $\cS$,
 \item there exist $m\geq m_1\geq m_2\geq\cdots\geq m_k\geq 1$ and $0=t_0\leq t_1\leq\cdots\leq t_k$ such that job $j_i$, $i\in\{1,\ldots,k\}$, executes non-preemptively in time interval $(t_{i-1},t_i)$ on each shared processor $\Mshared_l$, $l\in\{1,\ldots,m_i\}$, and in $(0,t_i)$ on $\Mpriv_{j_k}$  in $\cS$.
\end{itemize}
Figure~\ref{fig:example} gives an example of a synchronized schedule of five jobs where $m_1=4$, $m_2=m_3=3$, and $m_4=m_5=1$, and $t_1=1$, $t_2=2$,
$t_3=3.5$, $t_4=7$, and $t_5=10$. Observe that the total weighted overlap equals
\[\tct{\cS}=\sum_{i=1}^{k}(t_i-t_{i-1})\left(m_i\w{j_i}-\sum_{l=1}^{m_i}\cm{l}\right),\]
for a synchronized $\cS$, and by the feasibility of $\cS$
\[t_i+m_i(t_i-t_{i-1})=p_{j_i}.\]
From these two we get a natural interpretation of the objective function for synchronized $\cS$ which is summarized in the following formula
\[\tct{\cS}=\sum_{i=1}^{k}(p_{j_i}-t_{i})\left(\w{j_i}-\frac{\sum_{l=1}^{m_i}\cm{l}}{m_i}\right),\]
where $p_{j_i}-t_{i}$ is the \emph{total} time $j_i$ is executed on the $m_i$ cheapest shared processors $\Mshared_1,\ldots,\Mshared_{m_i}$, and the $(\sum_{l=1}^{m_i}\cm{l})/m_i$ is the \emph{average} cost of execution on these processors which must not exceed the weight, $w_{j_i}$, of the job $j_i$ if $\cS$ is to be optimal. Our main structural result in this paper is as follows.

\begin{theorem} \label{thm:synchronized}
There always exists an optimal schedule that is synchronized.
\end{theorem}

The theorem follows immediately from Corollary \ref{cor:segments} and Lemma \ref{lem:no-splits} in Section~\ref{sec:structure}. Though our algorithmic results, both optimization and approximation, do not require the schedule synchronization property directly, the property comes useful to prove some key properties of the algorithms. Therefore we start with synchronization here but leave technical details of its proof for  Section~\ref{sec:structure}.

The existence of optimal schedules that are synchronized is consistent with earlier results~\cite{VairaktarakisAydinliyim07, DK16,DK17,HK15} for the \emph{SP} job mode. However, the proof for the \emph{MP} job mode turns out to be more challenging. 
We note that both \emph{SP} and \emph{MP} modes are equivalent in the single processor case, $m=1$.

The computational complexity status of the problem remains a challenging open question. However, we show in Section \ref{sec:LP} that for any given permutation of job completions on private processors, recall that by definition there is exactly one job on any private processor,
an LP can be formulated that finds a schedule maximizing the total weighted overlap among all schedules that respect this permutation. 
This result points at a certain strategy in solving the problem which is to search for the duration each job must be executed on its private processor. The difficulty in establishing computational complexity status for the problem however indicates that the optimal durations will be also difficult to find. We therefore propose to \emph{relax} the strategy by requiring that each duration be at \emph{least} $\alpha>1/2$  fraction of  job processing time, i.e., we require that the completion of a job $j$ is limited to occur in the interval $[\alpha \p{j}, p_j]$. We naturally refer to such schedules as \emph{$\alpha$-private} since they require that at least $\alpha$ fraction of each job is executed of its private processor. The optimal $\alpha$-private schedules can be found in strongly polynomial time by solving a minimum-cost network flow problem, see Orlin \cite{O93} for a strongly polynomial algorithm for the minimum-cost network flow problem. These schedules are important since they guarantee that their total weighted overlaps are \emph{not less} than $\alpha$ of the maximum total weighted overlaps which results  in a $\alpha$-approximation strongly polynomial time algorithm for the problem. This is shown in Section~\ref{sec:approximation}, where we also show that the best $\alpha$ we can find in this paper is $\frac{1}{2}+\frac{1}{4(m+1)}$. The $\alpha$-approximation algorithm improves
a $1/2$-approximation algorithm for a single processor~\cite{DK17} to a $5/8$-approximation algorithm for $m=1$. 

Section~\ref{sec:antithetical} considers  \emph{antithetical} instances where $\p{j}\geq\p{j'}$ implies $\w{j}\leq\w{j'}$ for each pair of jobs $j$ and $j'$.
We prove that a permutation of job completions on private processors that coincides with non-decreasing order of processing times (or equivalently with non-increasing order of weights) is optimal
for the antithetical instances. The LP of Section \ref{sec:LP} then finds an optimal schedule for the permutation which gives a polynomial time algorithm for the antithetical instances.
Note that even though the permutation of job completions on private processors is fixed, the LP still needs to optimally decide how many shared processors should a job use. In particular, generally using more of them may decrease the job's contribution to the total weighted overlap of the schedule since the job uses more costly shared processors, on the other hand using more shared processors may reduce
the job completion time thus allowing jobs that follow it in the permutation to start earlier and thus to contribute proportionally more to the total weighted overlap of the schedule. 
Figure~\ref{fig:example}  illustrates this tradeoff for an antithetical instance, for example job 1 executed on all 4 shared processors finishes at $t_1=1$ and contributes 32 to the total weighted overlap, the same job executed on the cheapest three shared processors would finish later at $t_1=\frac{5}{4}$ however it would contribute $33\frac{3}{4}$ to  the total weighted overlap. The optimal solution for the instance choses the former which speeds up starting time of the remaining four jobs by $\frac{1}{4}$ in comparison to the latter. 
Generally this tradeoff is optimally handled by the LP, however it remains open whether an optimal schedule can be found by a more efficient or strongly polynomial time algorithm different from LP.

\section{An LP Formulation} \label{sec:LP}

In this section we give an LP formulation that takes as an input an instance of the problem with $n$ jobs $\jobs=\{1,\ldots,n\}$ having weights $w\colon\jobs\to\reals_+$, processing times $p\colon\jobs\to\reals_+$,  $m$ shared processors with costs $\cm{1},\ldots,\cm{m}$, and a permutation $A=(1,\ldots,n)$ of the jobs in $\jobs$. The permutation $A$ provides an order according to which the jobs finish on their private processors. The LP finds a schedule that maximizes the total weighted overlap among all $A$-compatible schedules.
We say that, for a permutation of jobs $A=(1,\ldots,n)$, a schedule $\cS$ is \emph{$A$-compatible} if $\complTime{\cS}{j}{\Mpriv}\leq\complTime{\cS}{j+1}{\Mpriv}$ for each $j\in\{1,\ldots,n-1\}$.

The variables used in the LP are as follows.
For each $i\in\{1,\ldots,n\}$, a variable $t_j$ is the completion time of job $j$ on its private processor.
For each $j\in\{1,\ldots,n\}$, $k\in\{1,\ldots,j\}$ and $i\in\{1,\ldots,m\}$, a variable $x_{jik}$ is the total amount of job $j$ executed on the shared processor $\Mshared_i$ in the time interval $(t_{k-1},t_{k})$.

\medskip
The LP is as follows:
\begin{equation} \label{LP:objective+}
\textup{maximize} \quad f=\sum_{j=1}^{n}\sum_{i=1}^{m}\sum_{k=1}^{j}(\w{j} -\cm{i}) x_{j i k}
\end{equation}
subject to:
\begin{equation} \label{LP:t+}
t_{0} =0 \leq t_{1} \leq \cdots  \leq t_{n},
\end{equation}
\begin{equation} \label{LP:exclusion+}
\sum_{j=k}^{n}x_{j i k} \leq t_{k} -t_{k-1}, \quad i\in\{1,\ldots,m\}, k\in\{1,\ldots,n\},
\end{equation}
\begin{equation} \label{LP:completion+}
\sum_{i=1}^{m} \sum_{k=1}^{j} x_{j i k} = \p{j} -t_{j}, \quad j\in\{1,\ldots,n\},
\end{equation}
\begin{equation} \label{LP:non-negative+}
x_{j i k}\geq 0, \quad j\in\{1,\ldots,n\}, k\in\{1,\ldots,j\}, i\in\{1,\ldots,m\}.
\end{equation}

For a solution to the  LP, we define the following \emph{corresponding} schedule $\cS$.
For each job $j\in\{1,\ldots,n\}$, let $\complTime{\cS}{j}{\Mpriv}=t_j$.
For each $k\in\{1,\ldots,j\}$ and $i\in\{1,\ldots,m\}$, execute a piece of job $j$ of duration $x_{j i k}$ on shared processor $\Mshared_i$ in time interval $(t_{k-1},t_{k})$.
These job pieces are executed in $(t_{k-1},t_k)$ on $\Mshared_i$ so that there is no overlap between them.

\begin{lemma} \label{lem:LP}
For each feasible solution to LP  the corresponding schedule $\cS$ is feasible, $A$-compatible and such that $f=\tct{\cS}$, and for each feasible and $A$-compatible schedule $\cS$  there is a feasible solution to LP such that $\tct{\cS}=f$.
\end{lemma}
\begin{proof}
Let $x_{j i k}$, $j\in\{1,\ldots,n\}$, $k\in\{1,\ldots,j\}$, $i\in\{1,\ldots,m\}$ and $t_{k}$, $k\in\{1,\ldots,n\}$ be a solution to the LP.  We first prove that a corresponding schedule $\cS$ is feasible and $A$-compatible.
For each $k\in\{1,\ldots,n\}$ and $i\in\{1,\ldots,m\}$, executing a piece of job $j$ of duration $x_{j i k}$ on shared processor $\Mshared_i$ in time interval $(t_{k-1},t_{k})$  is feasible since~\eqref{LP:exclusion+} ensures that the duration $x_{j i k}$ does not exceed the length of the interval. 
Also by~\eqref{LP:exclusion+} and by~\eqref{LP:non-negative+}, the length of the interval $(t_{k-1},t_{k})$ is sufficient to execute all job pieces of length $x_{j i k}$, $j\in\{1,\ldots,n\}$, on each shared processor $\Mshared_i$. Hence (\ref{L2}) is satisfied by $\cS$.
By~\eqref{LP:completion+}, the total execution time of all pieces of a job $j$ on all shared processors equals $\p{j}-t_j$, for each job $j\in\{1,\ldots,n\}$.
Since $\complTime{\cS}{j}{\Mpriv}=t_j$ in $\cS$, we obtain that the total length of all pieces of $j$ in $\cS$ equals $\complTime{\cS}{j}{\Mpriv}+\p{j}-t_j=\p{j}$ as required by (\ref{L1}).
This proves that $\cS$ is feasible and~(\ref{LP:t+}) implies that it is $A$-compatible. Finally $f$ in (\ref{LP:objective+}) equals the total weighted overlap in (\ref{L4}) since it can be readily verified that $\sum_{k=1}^jx_{jik}=\overlap{\cS}{j,\Mshared_i}$ for each $i\in\{1,\ldots,m\}$.

Now for an $A$-compatible feasible schedule $\cS$, set $t_j=\complTime{\cS}{j}{\Mpriv}$ for each $j\in\jobs$ and set $x_{j i k}$ to be the total execution time of job $j\in\jobs$ on shared processor $\Mshared_i$, $i\in\{1,\ldots,m\}$, in time interval $(t_{k-1},t_k)$, $k\in\{1,\ldots,n\}$, where $t_0=0$.
Since $\cS$ is $A$-compatible,~\eqref{LP:t+} is satisfied.
The constraint~\eqref{LP:exclusion+} is satisfied since  (\ref{L2}) holds in a feasible $\cS$. 
The constraint~\eqref{LP:completion+} is satisfied since by (\ref{L1}) the total execution time of each job in a feasible $\cS$ equals its processing time.
Finally,~\eqref{LP:non-negative+} follows directly from the definition of $x_{j i k}$'s. Thus the solution is a feasible solution to LP, and  the total weighted overlap of $\cS$ equals $f$ in~\eqref{LP:objective+}.
\end{proof}

\begin{theorem} \label{thm:A-compatible}
Given a permutation $A$ of jobs, a feasible and $A$-compatible schedule that maximizes the total weighted overlap can be computed in polynomial time.
\end{theorem}
\begin{proof}
By Lemma \ref{lem:LP}, the schedule $\cS$ corresponding to an optimal solution to LP is feasible and $A$-compatible, and it maximizes the total weighted overlap.
The optimal solution to LP can be found in polynomial time, see for instance~\cite{Sch}.
\end{proof}

\section{Approximation Algorithm} \label{sec:approximation}

In this section we show a strongly polynomial $\alpha$-approximation algorithm for the problem. The idea is to find a natural class of schedules for each instance of the problem such that optimal schedules in the class can be found in strongly polynomial time and such that those optimal schedules guarantee the required approximation $\alpha$. 

To that end we limit ourselves to \emph{$\alpha$-private} schedules in this section. In an $\alpha$-private schedule each job $j$ executes for at least $\alpha\p{j}$ time on its private processor $\Mpriv_j$, and in time interval $(0,\alpha\p{j})$ only on shared processors.
Although the optimal order of completion times of jobs on their \emph{private} processors in $\alpha$-private schedules is still difficult to find, and thus the LP from Section \ref{sec:LP} may not be used to find an optimal $\alpha$-private, we use another key property of those schedules which is that each job $j$ completes by $\alpha\p{j}$ on all \emph{shared} processors in formulating another linear program, we call it LA, to find optimal $\alpha$-private schedules in this section.  Therefore, each job $j$ must complete by $\alpha p_j$, $j\in\{1,\ldots,n\}$, on all shared processors in an $\alpha$-private schedule, and thus the order of these \emph{limiting} time points is clearly the same as the order of job processing times $\p{1}\leq\cdots\leq\p{n}$.
The completion time of $j$  on its private processor $\Mpriv_j$ will be set to $\complTime{\cS}{j}{\Mpriv}=(1-\alpha)\p{j}+\tilde{t}_j$ in $\alpha$-private schedules, where $\tilde{t}_j\geq 0$ is referred to as the \emph{remainder} of $j$. The choice of $\alpha$ needs to guarantee that $\alpha p_j\leq (1-\alpha)\p{j}+\tilde{t}_j$ so that each piece of job $j$ on shared processors counts for an overlap in an $\alpha$-private schedule. This inequality imposes an upper bound on $\alpha$. On the other hand we wish $\alpha$ to be as large as possible, in particular greater than a half, to guarantee as good as possible an approximation offered by $\alpha$-private schedules. This imposes a lower bound on $\alpha$. The compromise used in our LA is $\alpha=\frac{2m+3}{4(m+1)}=\frac{1}{2}+\frac{1}{4(m+1)}$. It remains open whether a higher value of $\alpha$ that meets both conditions can be found. We are now ready to show that the LA with this alpha finds an optimal $\alpha$-private schedule, i.e. an $\alpha$-private schedule that maximizes the total weighted overlap among all $\alpha$-private schedules.
The variables used in the LA are as follows.
For each $j\in\{1,\ldots,n\}$, a variable $\tilde{t}_j$ is the reminder of job $j$ to be executed on its private processor $\Mpriv_j$.
For each $j\in\{1,\ldots,n\}$, $k\in\{1,\ldots,j\}$ and $i\in\{1,\ldots,m\}$, a variable $x_{jik}$ is the total amount of job $j$ to be executed on the shared processor $\Mshared_i$ in the time interval $(\alpha \p{k-1},\alpha \p{k})$. 
We take $\p{0}=0$, and assume the order $(1,\ldots,n)$, $\p{1}\leq\cdots\leq\p{n}$ in the program.

\medskip
The LA is as follows:
\begin{equation} \label{LP:objective}
\textup{maximize} \quad \sum_{j=1}^{n}\sum_{i=1}^{m}\sum_{k=1}^{j}(\w{j} -\cm{i}) x_{j i k}
\end{equation}
subject to:
\begin{equation} \label{LP:t}
%t_{0} =0 \leq t_{1} \leq \cdots  \leq t_{n},
\frac{p_{j}}{2(m+1)}\leq \tilde{t}_j \leq \alpha\p{j}, \quad j\in\{1,\ldots,n\},
\end{equation}
\begin{equation} \label{LP:exclusion}
\sum_{j=k}^{n}x_{j i k} \leq \alpha(\p{k} - \p{k-1}), \quad i\in\{1,\ldots,m\}, k\in\{1,\ldots,n\},
\end{equation}
\begin{equation} \label{LP:completion}
\sum_{i=1}^{m} \sum_{k=1}^{j} x_{j i k} = \alpha\p{j} - \tilde{t}_{j}, \quad j\in\{1,\ldots,n\},
\end{equation}
\begin{equation} \label{LP:non-negative}
x_{j i k}\geq 0, \quad j\in\{1,\ldots,n\}, k\in\{1,\ldots,j\}, i\in\{1,\ldots,m\}.
\end{equation}

For a feasible solution to the  LA, we define the following \emph{corresponding} schedule $\cS$.
For each job $j\in\{1,\ldots,n\}$, set the completion time of $j$ on its private processor $\Mpriv_j$ to
\begin{equation} \label{eq:corresp}
\complTime{\cS}{j}{\Mpriv}=(1-\alpha) \p{j} + \tilde{t}_j.
\end{equation}
For each $k\in\{1,\ldots,n\}$ and $i\in\{1,\ldots,m\}$, execute a piece of job  $j$ of duration $x_{j i k}$ on shared processor $\Mshared_i$ in time interval $(\alpha\p{k-1},\alpha\p{k})$ in such a way that no two job pieces overlap. 
We now prove that $\cS$ is feasible.
\begin{lemma} \label{lem:corresp}
For a feasible solution to LA,
the corresponding schedule $\cS$ is feasible, and the value of objective function of the solution 
equals the total weighted overlap of $\cS$.
\end{lemma}
\begin{proof}
Let $\cS$ be a  schedule corresponding to a solution $x_{j i k}$, $j\in\{1,\ldots,n\}$, $k\in\{1,\ldots,j\}$, $i\in\{1,\ldots,m\}$ and $\tilde{t}_{j}$, $j\in\{1,\ldots,n\}$.
For each $k\in\{1,\ldots,j\}$ and $i\in\{1,\ldots,m\}$, executing a piece of job $j$ of duration $x_{j i k}$ on shared processor $\Mshared_i$ in time interval $(\alpha\p{k-1},\alpha \p{k})$  is feasible since~\eqref{LP:exclusion} guarantees that the duration $x_{j i k}$ does not exceed the length of the interval, and by~\eqref{LP:non-negative} the duration $x_{j i k}$ is non-negative.
Moreover, again by~\eqref{LP:exclusion} and~\eqref{LP:non-negative}, the length of the interval $(\alpha \p{k-1},\alpha \p{k})$ is sufficient to execute all pieces of jobs of length $x_{j i k}$, $j\in\{1,\ldots,n\}$, on each shared processor $\Mshared_i$.
By~\eqref{LP:completion}, the total execution time of all pieces of $j$ on all shared processors equals $\alpha\p{j}-\tilde{t}_j$, for each job $j\in\{1,\ldots,n\}$.
Since $\complTime{\cS}{j}{\Mpriv}=(1-\alpha)\p{j}+\tilde{t}_j$, we obtain that the total length of all pieces of $j$ in $\cS$ equals $\p{j}$ as required.
Moreover,  by~\eqref{LP:exclusion} and~\eqref{LP:completion} all the pieces of job $j$ that execute on shared processors end by $\alpha \p{j}$, and thus they end by 
$\complTime{\cS}{j}{\Mpriv}=(1-\alpha)\p{j}+\tilde{t}_j$ since by~\eqref{LP:t}, $\tilde{t}_j\geq p_j/(2(m+1))$, and $\alpha=\frac{1}{2}+\frac{1}{4(m+1)}$. This proves that $\cS$ is feasible.
Finally, by~\eqref{eq:corresp}, each job piece of $j$ on shared processor ends by its completion on the private one which shows that (\ref{LP:objective}) is the total weighted overlap of $\cS$.
\end{proof}

We now show that an optimal solution to the LA gives an $\alpha$-approximation of the optimum.

\begin{lemma} \label{lem:1/2}
For each input instance, the schedule $\cS$ that corresponds to an optimal solution to LA satisfies $\tct{\cS}\geq \alpha \tct{\Sopt}$, where $\Sopt$ is an optimal solution.
\end{lemma}
\begin{proof}
Suppose that an input instance consists of $n$ jobs $\{1,\ldots,n\}$ with processing times $\p{1}\leq\cdots\leq\p{n}$, weights $\w{1},\ldots,\w{n}$ and $m$ shared processors with costs $\cm{1}\leq\cdots\leq\cm{m}$. Let $\Sopt$ be an optimal schedule for the instance. By Theorem \ref{thm:synchronized} we may assume synchronized $\Sopt$.

For  $\Sopt$, let $y_{j i k}$ be equal to the total execution time of job $j$ on shared processor $i$ in time interval $(\p{k-1},\p{k})$ for each $j\in\{1,\ldots,n\}$, $i\in\{1,\ldots,m\}$ and $k\in\{1,\ldots,j\}$, where $\p{0}=0$.
Denote for brevity
\[e_j=\sum_{i=1}^{m}\sum_{k=1}^{j} y_{j i k}\]
to be the total amount of  job $j$ executed on all shared processors in $\Sopt$. From the synchronization we observe 
\begin{equation} \label{b}
0\leq e_j\leq \frac{mp_j}{(m+1)}.
\end{equation}
We assign values to the variables in the LA as follows:
\begin{align}
x_{j i k}   & = \alpha y_{j i k}, \quad j\in\{1,\ldots,n\}, k\in\{1,\ldots,j\},i\in\{1,\ldots,m\}, \label{eq:Sopt1} \\
\tilde{t}_j & = \alpha(\p{j} - e_j), \quad j\in\{1,\ldots,n\}. \label{eq:Sopt2}
\end{align}
We prove that this assignment gives a feasible solution to the LA.
%Since $\Sopt$ is feasible, we have $0\leq e_j\leq\p{j}$, which
By (\ref{b}) and $\alpha>\frac{1}{2}$, we have~\eqref{LP:t} satisfied.
For each shared processor $\Mshared_i$ and each interval $(\p{k-1},\p{k})$, the total execution time of all job pieces executed in this interval on $\Mshared_i$ in the schedule $\Sopt$ is
$\sum_{j=1}^ny_{j i k}$.
Thus, since $\Sopt$ is feasible (and in particular we use the fact that no job $j$ executes on a shared processor after time point $\p{j}$ in $\Sopt$), we have for each $i\in\{1,\ldots,m\}$
\[\sum_{j=k}^ny_{j i k}\leq \p{k}-\p{k-1}.\]
This proves, by~\eqref{eq:Sopt1}, that~\eqref{LP:exclusion} holds.
The total amount of a job $j$ that executes on all shared processors in $\Sopt$ is $e_j$ and hence 
\[e_j = \sum_{i=1}^m\sum_{k=1}^j y_{j i k} = \frac{1}{\alpha}\sum_{i=1}^m\sum_{k=1}^j x_{j i k},\]
which by~\eqref{eq:Sopt2} gives~\eqref{LP:completion}.
Finally,~\eqref{LP:non-negative} follows directly from~\eqref{eq:Sopt1} and the feasibility of $\Sopt$.

Let $\cS$ be the schedule corresponding to the above LA solution.
By Lemma~\ref{lem:corresp}, $\cS$ is indeed a feasible schedule, and
the total weighted overlap of $\cS$ equals by~\eqref{eq:Sopt1}:
\[\tct{\cS}=\sum_{j=1}^n\sum_{i=1}^m\sum_{k=1}^j x_{j i k}(\w{j}-\cm{i}) = \sum_{j=1}^n\sum_{i=1}^m\sum_{k=1}^j \alpha y_{j i k}(\w{j}-\cm{i}) = \alpha \tct{\Sopt},\]
which completes the proof since $\tct{\cS^*}\geq \tct{\cS}$, where $\cS^*$ is the schedule that corresponds to an optimal solution to the LA.
\end{proof}

We now argue that the LA can be recast as a minimum-cost network flow problem which can be solved more efficiently than a general linear program, see Orlin \cite{O93}.
The directed flow network $D=(V,A)$ can be constructed as follows.
For each $i\in\{1,\ldots,m\}$ and $k\in\{1,\ldots,n\}$ introduce two nodes $v_{ik}$, $v_{ik}'$.
For each job $j\in\{1,\ldots,n\}$, introduce a node $u_j$, and let $s$ and $t$ be the source and sink nodes in the network.
We add the following arcs to $D$:
\begin{enumerate}[label={\normalfont{(\alph*)}}]
 \item for each $j\in\{1,\ldots,n\}$, let $(s,u_j)\in A$ be an arc of capacity $\arcCap{(s,u_j)}=\frac{m\p{j}}{2(m+1)}$ and cost $\arcCost{(s,u_j)}=0$,
 \item for each $i\in\{1,\ldots,m\}$, $j\in\{1,\ldots,n\}$ and $k\in\{1,\ldots,j\}$,  let $(u_j,v_{ik})\in A$ be an arc of capacity $\arcCap{(u_j,v_{ik})}=+\infty$ and cost $\arcCost{(u_j,v_{ik})}=\w{j}-\cm{i}$,
 \item for each $i\in\{1,\ldots,m\}$ and $k\in\{1,\ldots,n\}$, let $(v_{ik},v_{ik}')\in A$ be an arc of capacity $\arcCap{(v_{ik},v_{ik}')}=\alpha(\p{k}-\p{k-1})$ and cost $\arcCost{(v_{ik},v_{ik}')}=0$,
 \item for each $i\in\{1,\ldots,m\}$ and $k\in\{1,\ldots,n\}$, let $(v_{ik}',t)\in A$ be an arc of capacity $\arcCap{(v_{ik}',t)}=+\infty$ and cost $\arcCost{(v_{ik}',t)}=0$.
\end{enumerate}
Suppose that $f$ is an $s$-$t$ flow in $D$.
We define a feasible LA solution by taking $x_{jik}=f(u_j,v_{ik})$ for each $j\in\{1,\ldots,n\}$, $i\in\{1,\ldots,m\}$, $k\in\{1,\ldots,j\}$ and $\tilde{t}_j=\alpha p_j-f(s,u_j)$ for each $j\in\{1,\ldots,n\}$.
Constraint~\eqref{LP:t} is satisfied because $\arcCap{(s,u_j)}=\frac{m\p{j}}{2(m+1)}$.
Constraint~\eqref{LP:exclusion} follows directly from $\arcCap{(v_{ik},v_{ik}')}=\alpha(\p{k}-\p{k-1})$ since the flow through the arc $(v_{ik},v_{ik}')$ equals $\sum_{j=k}^n f(u_j,v_{ik})=\sum_{j=k}^n x_{jik}$.
For~\eqref{LP:completion} we have for each $j\in\{1,\ldots,n\}$,
\[\sum_{i=1}^m\sum_{k=1}^j x_{jik} = \sum_{i=1}^m\sum_{k=1}^j f(u_j,v_{ik}) = f(s,u_j) = \alpha \p{j} - \tilde{t}_j. \]
Then,~\eqref{LP:non-negative} is due to the fact that the flow is non-negative.
Since all arcs except for $(u_j,v_{ik})$'s have cost zero, we obtain that the objective function in~\eqref{LP:objective} equals the total cost of the flow $f$.
The construction leading from an LA solution to a flow in $D$ is straightforward and analogous and hence we skip its description.

We have proved the following.
\begin{theorem} \label{thm:approximation}
For any input instance with $m\geq 1$ processors, there exists a strongly polynomial approximation algorithm with approximation ratio $\alpha=\frac{1}{2}+\frac{1}{4(m+1)}$.
\qed
\end{theorem}

Observe that by  Lemma \ref{lem:1/2} the LA is $\frac{5}{8}$- approximation for a single shared processor problem, $m=1$, which improves the  $\frac{1}{2}$- approximation in  \cite{DK17} but at a cost of computational complexity which however still remains strongly polynomial.

\section{Processor-descending and sequential schedules} \label{sec:segments}

As a stepping stone towards the proof  of Theorem \ref{thm:synchronized} and towards the algorithm for antithetical instances we show that there always exist optimal schedules that are \emph{processor-descending} and \emph{sequential} in this section. 

For a given schedule $\cS$, we say that an interval $I$ is a \emph{segment} in $\cS$ if $I$ is a maximal interval such that each shared processor is either idle in $I$ or has no idle time in $I$.
We say that a job is \emph{present} in a segment if some non-empty part of the job executes in this segment.
We then say that a segment $I$ is \emph{sequential} if each job $j$ is either not present in $I$ or there exists an interval $I'\subseteq I$, called the \emph{interval of $j$ in $I$}, such that each processor that is not idle in $I$ executes $j$ exactly in the interval $I'$.
If each segment in a schedule is sequential, then the schedule is called \emph{sequential}.
We say that a schedule is \emph{processor-descending} if each shared processor has no idle times between any two job pieces it executes, and for any two processors $\Mshared_i$ and $\Mshared_{i'}$ with $\cm{i}<\cm{i'}$ it holds that $\Mshared_i$ completes executing all job pieces not later than $\Mshared_{i'}$. A job $j$ is \emph{synchronized} in a processor-descending and sequential schedule $\cS$ if the last piece of $j$ ends on shared processors at $\complTime{\cS}{j}{\Mpriv}$.

We now describe a simple schedule modification that, without increasing the total weighted overlap, arrives at a schedule that is processor-descending and sequential (see Figure~\ref{fig:segment} for an illustration).
We name this transformation as a procedure as it will be used later.

\procStart
\textbf{Procedure} $\procSequential(\cS)$\\
\textbf{Input:} A feasible schedule $\cS$.\\
\textbf{Output:} A processor-descending and sequential schedule $\cS'$ with $\tct{\cS'}\geq\tct{\cS}$.
\begin{enumerate} [label={\normalfont{(M\arabic*)}},nosep]
 \item\label{it:s1} As long as there exists a processor $\Mshared_i$ and an idle time $(a,b)$ (take this idle time to have maximum duration) followed by a piece $(a',b')$ of a job $j$ do the following: move the piece of $j$ to be executed in time interval $(a,a+b'-a')$.
 \item\label{it:s2} For each segment $I$ in $\cS$ do the following:
  \begin{enumerate} [label={\normalfont{(M\arabic{enumi}\alph*)}},nosep]
   \item\label{it:s2a} Let $j_1^I,\ldots,j_{l(I)}^I$ be the jobs present in segment $I$ sorted according to non-decreasing order of their completion times on private processors, $\complTime{\cS}{j_1^I}{\Mpriv}\leq\cdots\leq\complTime{\cS}{j_{l(I)}^I}{\Mpriv}$. Let $a_i^I$ be the total amount of job $j_i^I\in\{j_1^I,\ldots,j_{l(I)}^I\}$ executed in the segment $I$. Let $m'$ be the number of processors used by $I$ in $\cS$.
   \item\label{it:s2b} Replace the segment $I$ in $\cS$ with one in which the job $j_{t}^I$, $t\in\{1,\ldots,l(I)\}$, executes in time interval
    \[\left( L + \frac{1}{m'}\sum_{t'=1}^{t-1}a_{t'}^I, L + \frac{1}{m'}\sum_{t'=1}^{t}a_{t'}^I \right)\]
    on all these $m'$ shared processors, where $L$ is the left endpoint of $I$.
  \end{enumerate}
\end{enumerate}
\procEnd

\begin{figure}[ht!]
\begin{center}
 \includegraphics[scale=1]{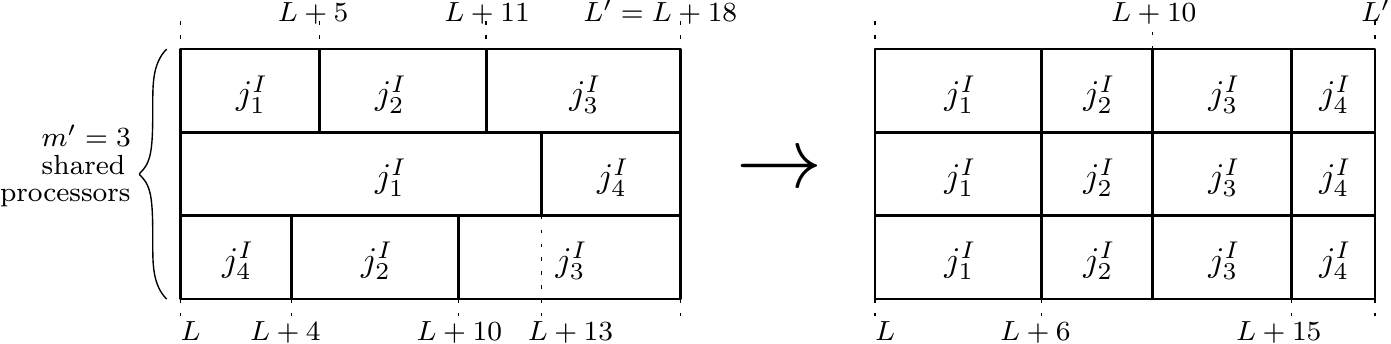}
\end{center}
\caption{Transformation performed in Step~\ref{it:s2} for a segment $I=(L,L')$ with total job executions times in this segment being $a_1^I=18$, $a_2^I=12$, $a_3^I=15$, $a_4^I=9$}
\label{fig:segment}
\end{figure}

\begin{lemma} \label{lem:segments}
Procedure~$\procSequential$ transforms an input schedule $\cS$ into a schedule that is processor-descending, sequential and has the same total weighted overlap.
\end{lemma}
\begin{proof}
Note that Step~\ref{it:s1} of Procedure~$\procSequential$ makes $\cS$ to be processor-descending because after this transformation, each shared processor is busy in a single time interval that starts at $0$.
Recall that, if in a given segment of the new schedule it holds that $m'$ shared processors are used, then since they are ordered according to their costs, we may without loss of generality assume that these processors are $\Mshared_1,\ldots,\Mshared_{m'}$.

Consider now Step~\ref{it:s2} of Procedure~$\procSequential$.
The fact that the total execution time of each job on shared processors does not change within the segment follows directly from the formula in Step~\ref{it:s2b}.
We need to prove that for each $t\in\{1,\ldots,l(I)\}$, the job $j_t^I$ completes its piece in segment $I$ in the output schedule $\cS'$ not later than its completion time on its private processor:
\begin{equation} \label{eq:filled}
\complTime{\cS'}{j_t^I}{\Mpriv} \geq L + \frac{1}{m'}\sum_{t'=1}^{t}a_{t'}^I=e_t
\end{equation}
since this implies that $\tct{\cS'}=\tct{\cS}$.
%Denote the expression on the right in the above inequality by $e_t$.
For each $t\in\{1,\ldots,l(I)\}$, if there exists a shared processor $\Mshared_i$ such that a piece of $j_t^I$ ends on $\Mshared_i$ in $\cS$  at $e_t$ or later, then we are done.
Suppose for a contradiction that~\eqref{eq:filled} does not hold, i.e., $\complTime{\cS'}{j_t^I}{\Mpriv}<e_t$ for some $t$. 
Since $\complTime{\cS}{j_t^I}{\Mpriv}=\complTime{\cS'}{j_t^I}{\Mpriv}$, we have that the job $j_t^I$ completes \emph{before} $e_t$ on each shared processor in $\cS$.
Thus, since there is no idle time in interval $(L,e_t)$ in $\cS'$, there exists $1\leq\bar{t}<t$ such that the job $j_{\bar{t}}^I$ completes on some shared processor at time $\bar{e}$, after the time point $e_t$ in $\cS$, $\bar{e}>e_t$.
But according to the job ordering picked in Step~\ref{it:s2a}, $\complTime{\cS}{j_{\bar{t}}^I}{\Mpriv}\leq\complTime{\cS}{j_t^I}{\Mpriv}$.
By assumption $\complTime{\cS}{j_t^I}{\Mpriv}=\complTime{\cS'}{j_t^I}{\Mpriv}<e_t$.
This gives that $\bar{e}>\complTime{\cS}{j_{\bar{t}}^I}{\Mpriv}$, which violates (\ref{L3}) and gives the required contradiction since $\cS$ is feasible.
This proves~\eqref{eq:filled} and completes the proof of the lemma.
\end{proof}

Thus we conclude.
\begin{corollary} \label{cor:segments}
There exists an optimal processor-descending and sequential schedule.
\qed
\end{corollary}

\section{Antithetical Instances} \label{sec:antithetical}

An instance $\jobs$ is \emph{antithetical} if for any two jobs $j$ and $j'$ it holds: $\p{j}\leq\p{j'}$ implies $\w{j}\geq\w{j'}$. Our main goal in this section is to show a polynomial time algorithm for antithetical instances. The algorithm relies on the LP given in Section \ref{sec:LP} which however requires an optimal job order to produce an optimal solution. We prove that the ascending order of
processing times is such an order, and that all jobs occur on shared processors in optimal schedules produced by the algorithm. The proof relies on a transformation, called $j$-filling, of a schedule
which we now define. The transformation may produce schedules which are not synchronized even if applied to a synchronized schedule, however those schedules must then be processor-descending and sequential with synchronized suffixes inherited from  the original synchronized schedule.

For a synchronized schedule $\cS$ we say that it is \emph{processing-time ordered} if $\p{j_1}\leq\cdots\leq\p{j_k}$, where $j_1,\ldots,j_k$ are the jobs that appear, in $\cS$ in this order, on the shared processors.
We then for brevity say that $(j_1,\ldots,j_k)$ is the \emph{ordering} of jobs in $\cS$.
Consider an arbitrary processor-descending and sequential schedule $\cS$.
We say that a suffix $(j_1,\ldots,j_k)$, is \emph{processing-time ordered} and \emph{synchronized} in $\cS$ if $\p{j_1}\leq\cdots\leq\p{j_k}$, there exists a time point $t$ such that exactly the jobs $j_1,\ldots,j_k$ execute in time interval $(t,+\infty)$ in this order, and the jobs $j_1,\ldots,j_k$ are synchronized.
The synchronization and processing-time ordering are thus not required for the entire schedule but only for some suffix of $\cS$.

Consider a processor-descending and sequential schedule $\cS$ such that there exists a job $j$ which satisfies one of the following.
\begin{enumerate}[label={\normalfont{(\arabic*)}}]
\item\label{it:antit1} $j$ is present on the shared processors, $j$ is not synchronized, and $j$ is followed by a processing-time ordered synchronized suffix $(j_1,\ldots,j_k)$ in $\cS$.
\item\label{it:antit2} $j$ is not present on the shared processors, and $\cS$ has a processing-time ordered synchronized suffix $(j_1,\ldots,j_k)$ that starts at time $t_1<\p{j}$.
\end{enumerate}

Let $t_1$ and $t_1'$ be such that the piece of the job $j_1$ executes in $(t_1,t_1')$ on the shared processors.
We define an operation of \emph{$j$-filling} in $\cS$ as follows (see Figure~\ref{fig:filling}(a)): for some $t$, $t_1<t\leq t_1'$, the part of $j_1$ executing in time interval $(t_1,t)$ is moved from each shared processor $\Mshared_z$, $z\in\{1,\ldots,m'\}$ to its private processor $\Mpriv_{j_1}$, and it is replaced on $\Mshared_z$ by  $j$ so that the completion time of $j$ on $\Mpriv_j$ decreases by $m'(t-t_1)$, where $m'$ is the number of shared processors used in the interval $(t_1,t_1')$.
The $t\in(t_1,t_2]$ is chosen to be maximum to ensure that the completion time of $j$ on shared processors, which equals $t$, is smaller than or equal to the completion time of $j$ on $\Mpriv_j$.
\begin{figure}[ht!]
\begin{center}
 \includegraphics[scale=1.2]{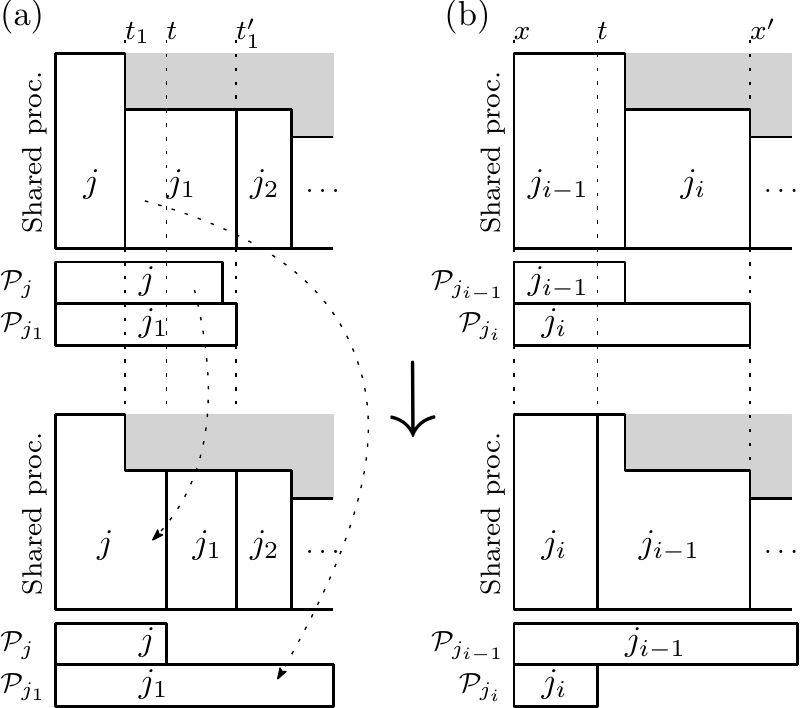}
\end{center}
\caption{(a) $j$-filling for the case when $j$ is present on the shared processors;
         (b) changing the order of jobs $j_{i-1}$ and $j_i$ on the shared processors when $\p{j_{i-1}}>\p{j_i}$}
\label{fig:filling}
\end{figure}
We remark that if $t<t_1'$, then the job $j$ becomes synchronized as a result of $j$-filling --- informally speaking this follows from observation that further increase of $t$ is not possible due to the fact that there is not enough of $j$ in time interval $(t_1,\complTime{\cS}{j}{\Mpriv})$ on private processor $\Mpriv_{j}$ in $\cS$ to fill out the interval $(t_1,t+\varepsilon)$ on the shared processors for any $\varepsilon>0$ (this case is depicted in Figure~\ref{fig:filling}(a)).
On the other hand, if $t=t_2$, then $j$ may not be synchronized  as a result of $j$-filling.

Note that the definition of $j$-filling is valid, it suffices to observe that the maximum $t$ selected indeed satisfies $t_1<t$.
This follows from the assumption that $\p{j}>t_1$ when $j$ is not present on the shared processors, and from the fact that $j$ is not synchronized and directly precedes $j_1$ on the shared processors otherwise.
Note that $\w{j}\geq\w{j_1}$ is sufficient to ensure that as a result of $j$-filling the total weighted overlap does not decrease in comparison to $\cS$.
Since the execution of jobs $j_{2},\ldots,j_k$ does not change as result of $j$-filling, we obtain:
\begin{observation} \label{obs:j-filling-feasible}
Suppose that a schedule $\cS$ and $j$ satisfy the assumptions~\ref{it:antit1} or~\ref{it:antit2} of $j$-filling, and $\w{j}\geq\w{j_1}$.
The operation of $j$-filling gives a feasible schedule $\cS'$ such that $\tct{\cS'}\geq\tct{\cS}$, the job $j_1$ is either not synchronized in $\cS'$ and present on the shared processors (this holds when $t<t_1'$), or the job $j_1$ executes on $\Mpriv_{j_1}$ only and $\complTime{\cS'}{j_1}{\Mpriv}>t_1'$ (this holds when $t=t_1'$) and the suffix $(j_{i+1},\ldots,j_k)$ is processing-time ordered and synchronized in $\cS'$.
\qed
\end{observation}

We are now ready to prove the main result of this section.
\begin{theorem} \label{thm:antithetical}
For any antithetical instance, there exists an optimal schedule that is processing-time ordered and each job is present on the shared processors.
Moreover, it can be computed in polynomial time.
\end{theorem}

\begin{proof}
Consider an optimal schedule $\cS$ for an antithetical instance.
By Theorem~\ref{thm:synchronized} we may assume that $\cS$ is synchronized.
Let $(j_1,\ldots,j_k)$ be the ordering of jobs in $\cS$.
We first argue that $\cS$ is processing-time ordered.
We prove this by contradiction: take the largest index $i$ such that $\p{j_{i-1}}>\p{j_i}$.
Swap the jobs $j_{i-1}$ and $j_i$ on the shared processors as follows (see Figure~\ref{fig:filling}(b)):
Suppose that $j_{i-1}$ and $j_i$ occupy the interval $(x,x')$ on the shared processors in $\cS$.
Find the $t$, $x<t\leq x'$, so that when replacing $j_{i-1}$ by $j_i$ in time interval $(x,t)$ on each shared processor and executing $j_i$ in $(0,t)$ on its private processor results in $j_i$ having the total execution time equal to $\p{j_i}$. Thus, $j_i$ remains synchronized.
Finally execute $j_{i-1}$ in time interval $(t,x')$ on each shared processor on which $j_{i-1}$ or $j_i$ was initially present, and execute the remainder of $j_{i-1}$ on its private processor.
The fact that $\p{j_{i-1}}>\p{j_i}$ implies that this swap gives a feasible schedule and that $j_{i-1}$ is no longer synchronized.
By the maximality of $j$, the suffix $(j_{i+1},\ldots,j_k)$ in the resulting schedule is processing-time ordered and synchronized (as nothing in the suffix has changed with respect to $\cS$).
Perform $j_{i-1}$-filling to $\cS$, and then for each $i':=i+1,\ldots,k-1$ (in this order) apply $j_{i'}$-filling, obtaining a final schedule $\cS'$.
By Observation~\ref{obs:j-filling-feasible}, $\cS'$ is feasible, $\tct{\cS'}\geq\tct{\cS}$ and the job $j_k$ is either not synchronized or not present on the shared processors in $\cS'$.
In both cases we obtain that $\cS'$ is not optimal (observe that $j_k$ completes later on its private processor than the last job completes on shared processors in $\cS'$ which is obviously not optimal since some part of $j_k$ can be moved to the cheapest shared processor and thus increase the overlap), which contradicts the optimality of $\cS$.
Thus, we have proved that $\cS$ is processing-time ordered.

We now prove that $\jobs=\{j_1,\ldots,j_k\}$, i.e, all jobs are present on the shared processors in $\cS$.
By contradiction, let $j\notin\{j_1,\ldots,j_k\}$.
If $\p{j}\geq\p{j_k}$, then $\cS$ is not optimal and we immediately obtain a contradiction.
Otherwise, since $\cS$ is processing-time ordered and synchronized, as we showed earlier in the proof, there is a suffix $(j_i,\ldots,j_k)$ of $\cS$ for which the condition~\ref{it:antit2} of $j$-filling is satisfied ($i$ is the maximum index such that $\p{j}>\p{j_i}$).
Perform the $j$-filling and then iteratively for $i':=i,\ldots,k-1$ (in this order) perform $j_{i'}$-filling obtaining the final $\cS'$.
Again by Observation~\ref{obs:j-filling-feasible}, $\cS'$ is feasible, $\tct{\cS'}\geq\tct{\cS}$ and $j_k$ is either not present on shared processors or is not synchronized, giving us the required contradiction.

Finally, the LP (see Theorem~\ref{thm:A-compatible}) gives the optimal processing-time ordered schedule in polynomial time.
\end{proof}

\section{Structure of Optimal Schedules} \label{sec:structure}

This section proves Theorem~\ref{thm:synchronized} that was announced earlier in the paper. By Corollary \ref{cor:segments} we can limit ourselves to processor-descending and sequential schedules $\cS$. Those schedules may not be synchronized for a number of reasons:  a job may appear in more than one segment of $\cS$, we call this a split of the job, or even if each job appears in at most one segment of $\cS$ some jobs may not be synchronized by finishing on shared processors earlier than on their private processors. We need to show how to remove these undesirable configurations from processor-descending and sequential schedules  to produce synchronized schedules without decreasing the total weighted overlap in the process. This removal affects schedules and their total weighted overlaps in a quite complicated way that requires sometimes delaying parts of the schedules whereas at other times their advancing in order not to reduce the total wighted overlap. We describe the main building block of the transformation, we call it \emph{modification}, and its key properties in the next subsection. The modification will be used in Subsection \ref{sec:split} to remove the splits, and in Subsection \ref{sec:synchronization} to synchronize jobs.

\subsection{Towards Schedule Synchronization} \label{sec:xi}

Let $\cS$ be a processor-descending and sequential schedule. Suppose  $\cS$ has $\ell$ segments, $S_1,\ldots,S_{\ell}$ and the $i$-th segment executes jobs $j_{i,1},\ldots,j_{i,l(i)}$ in time intervals $(s_{i,1},e_{i,1}),\ldots,(s_{i,l(i)},e_{i,l(i)})$, respectively. Define $T(\cS)=\{s_{i,k}, e_{i,k}: i=1,\ldots,\ell; k=1,\ldots,l(i)\}$ to be the set of all time points $t$ such that some piece of a job starts or ends at $t$ on a shared processor. Let $m_{i,k}$ be the number of shared processors used by $\cS$ in the interval $(s_{i,k},e_{i,k})$. Observe that this number remains the same for each interval in a segment $i$ thus we denote it by $m_i$ and we refer to it as the \emph{width} of the interval $(s_{i,k},e_{i,k})$.  We define the \emph{factor} $m^+_{i,k}$ and the \emph{radius} $r_{i,k}$ of the interval $(s_{i,k},e_{i,k})$ as follows.
Let a job $j_{i,k}$ execute in an interval $(s_{i,k},e_{i,k})$, if $e_{i,k}=\complTime{\cS}{j_{i,k}}{\Mpriv}$, then $m^+_{i,k}=m_i+1$ and $r_{i,k}=\min\{e_{i,k}-s_{i,k},p_{j_{i,k}}-e_{i,k}\}$. Otherwise, if $e_{i,k}<\complTime{\cS}{j_{i,k}}{\Mpriv}$, then  $m^+_{i,k}=m_i$ and $r_{i,k}=\min\{e_{i,k}-s_{i,k},\complTime{\cS}{j_{i,k}}{\Mpriv}-e_{i,k}\}$.

In this section we define a transformation of $\cS$ that would be used to make it synchronized.
The transformation is multi-step which in each step is defined as a function $\xi$ that takes a schedule $\cS$, a time point $t\in T(\cS)$ and a \emph{shift} 
$\varepsilon\in\reals$ as an input, and produces a schedule $\cS'$ and a new shift $\varepsilon'$ as output.

Before giving its formal description, we start with some informal intuitions.
We consider three basic steps that make up the whole transformation.
The first step is the \emph{base} step of our transformation: this case simply moves the endpoint of the last job piece of the entire schedule, i.e., the piece that ends at $e_{\ell,l(\ell)}$.
If $\varepsilon>0$, then this piece is moved to the right (i.e., it completes later in $\cS'$ than in $\cS$), and if $\varepsilon<0$, then this piece advances in $\cS'$ with respect to $\cS$.

The \emph{main} step  is subdivided  into two subcases.
In the \emph{first} subcase we consider a piece of a job $j$ that ends earlier on shared processors than on the private processor.
Note that if this is the last piece of  $j$, then it implies that $j$ is not synchronized.
However, it may also happen that this is not the last piece of $j$ but $j$ itself is synchronized as there may be another piece of $j$ in one of the subsequent segments of $\cS$.
In this situation the endpoint of the piece of $j$ is just moved (on each shared processor) to the right or to the left (according to whether $\varepsilon>0$ or $\varepsilon<0$, respectively).

In the \emph{second} subcase we consider the last piece of a job $j$ that is synchronized.
Then, we shift both the endpoint of the piece of $j$ on shared processors and, by the same amount, the endpoint on the private processor.
In this way the job remains synchronized.

For each step we define a payoff value that tells how much the total weighted overlap of the schedule changes by doing the step.
We need to keep in mind the multi-step nature of the entire transformation.
Typically the transformation starts with some $\cS$, time point $t=e_{i,b}$ and $\varepsilon$, then it will subsequently trigger changes for the same parameter $\cS$, subsequent time points
\[e_{i,b+1},\ldots,e_{i,l(i)},\quad\ldots\quad,e_{i+1,1},\ldots,e_{i+1,l(i+1)},\quad\ldots\quad,e_{\ell,1},\ldots,e_{\ell,l(\ell)},\]
and different values of $\varepsilon$.

We now give a description of these three steps and then an example that depicts all of them follows (see Figure~\ref{fig:xi}).
The changes introduced to $\cS$ in each of these steps are referred to as one step \emph{modifications}.

\noindent
\textbf{Base Step.}
The assumption of this step is that $t=e_{\ell,l(\ell)}$, i.e., $t$ is the end of the last job $j_{\ell,l({\ell})}$ of the last segment of $\cS$.
We call this job the \emph{job of the modification} and denote by $j$ for convenience.
We may assume without loss of generality that $j$ is synchronized and hence $t=\complTime{\cS}{j}{\Mpriv}$. 
The modification is \emph{doable} if
\begin{equation} \label{eq:trans1:base1}
\varepsilon\in\left[- m^+_{\ell,l(\ell)}r_{\ell,l(\ell)}, m^+_{\ell,l(\ell)}r_{\ell,l(\ell)}\right].
\end{equation}
For the doable modification, we set the completion time of  $j$ on processors $\Mshared_1,\ldots,\Mshared_{m_{\ell}}$ and $\Mpriv_j$ in $\cS'$ to
\begin{equation} \label{eq:trans1:base2}
e_{\ell,l(\ell)}+\frac{\varepsilon}{m^+_{\ell,l(\ell)}}.
\end{equation}
This transformation is denoted by $\xi(\cS,t,\varepsilon)$ and its \emph{payoff} equals
\begin{equation} \label{eq:trans1:base3}
\payoff{\xi(\cS,t,\varepsilon)} = \frac{\varepsilon}{m^+_{\ell,l(\ell)}} \sum_{z=1}^{m_{\ell}}(\w{j}-\cm{z}).
\end{equation}
This completes the description of the base step.

In the two remaining steps, we assume $t=e_{i,b}$ where $i<\ell$ or $b<l(\ell)$ and use some common notation for both.
Let for brevity $j=j_{i,b}$.
Note that the interval that immediately follows $(s_{i,b},e_{i,b})$ is either $(s_{i,b+1}=e_{i,b},e_{i,b+1})$ when $b<l(i)$, i.e., when  $j$ is not the last in the segment $S_i$ or $(s_{i+1,1}=e_{i,b},e_{i+1,1})$ when $b=l(i)$, i.e.,  $j$ is last in the segment $S_i$. Let 
$m'=m_i$ in the former case, and 
$m'=m_{i+1}$ in the latter case. Finally, let $j'$ be the job in the interval that immediately follows $(s_{i,b},e_{i,b})$.

\medskip
\noindent
\textbf{Main Step~I.}
The assumption of this step is that $e_{i,b}< \complTime{\cS}{j}{\Mpriv}$. 
We say that the modification is \emph{doable} if
\begin{equation} \label{eq:trans1:CaseI1}
\varepsilon\in\left[ -m^+_{i,b}r_{i,b}, m^+_{i,b}r_{i,b} \right],
\end{equation}
For the doable modification set the completion time of $j$ on processors $\Mshared_1,\ldots,\Mshared_{m_i}$ equal to the start time of  $j'$ on shared processors $\Mshared_1,\ldots,\Mshared_{m'}$ to 
\begin{equation} \label{eq:trans1:CaseI2}
t' = e_{i,b}+\frac{\varepsilon}{m^+_{i,b}},
\end{equation}
and denote $\nextEps{\varepsilon}=\varepsilon \frac{m'}{m^+_{i,b}}$.
The \emph{payoff} is
\begin{equation} \label{eq:trans1:CaseI3}
\payoff{\xi(\cS,t,\varepsilon)} = 
\frac{\varepsilon}{m^+_{i,b}} \left( \sum_{z=1}^{m_i}(\w{j}-\cm{z}) - \sum_{z=1}^{m'}(\w{j'}-\cm{z}) \right).
\end{equation}

\noindent
\textbf{Main Step~II.}
The assumption of this step is:
\begin{equation} \label{eq:trans1:CaseII1}
e_{i,b} = \complTime{\cS}{j}{\Mpriv}.
\end{equation}
We say that the modification is \emph{doable} if
\begin{equation} \label{eq:trans1:CaseII2}
\varepsilon\in\left[ -m^+_{i,b}r_{i,b}, m^+_{i,b}r_{i,b} \right]
\end{equation}
For the doable modification set the completion time of  $j$ on processors $\Mshared_1,\ldots,\Mshared_{m_i}$ and $\Mpriv_j$ equal to the start time of  $j'$ on shared processors $\Mshared_1,\ldots,\Mshared_{m'}$ to 
\begin{equation} \label{eq:trans1:CaseII3}
t' = e_{i,b}+\frac{\varepsilon }{m^+_{i,b}},
\end{equation}
and denote $\nextEps{\varepsilon}=\varepsilon \frac{m'}{m^+_{i,b}}$.
The \emph{payoff} is then
\begin{align} \label{eq:trans1:CaseII4}
\begin{split}
\payoff{\xi(\cS,t,\varepsilon)} = 
 \frac{\varepsilon}{m^+_{i,b}} \left( \sum_{z=1}^{m_i}(\w{j}-\cm{z}) - \sum_{z=1}^{m'}(\w{j'}-\cm{z}) \right).
\end{split}
\end{align}

This completes the description of all cases of our transformation --- see Figure~\ref{fig:xi} for an example.
\begin{figure}[ht!]
\begin{center}
 \includegraphics[scale=1]{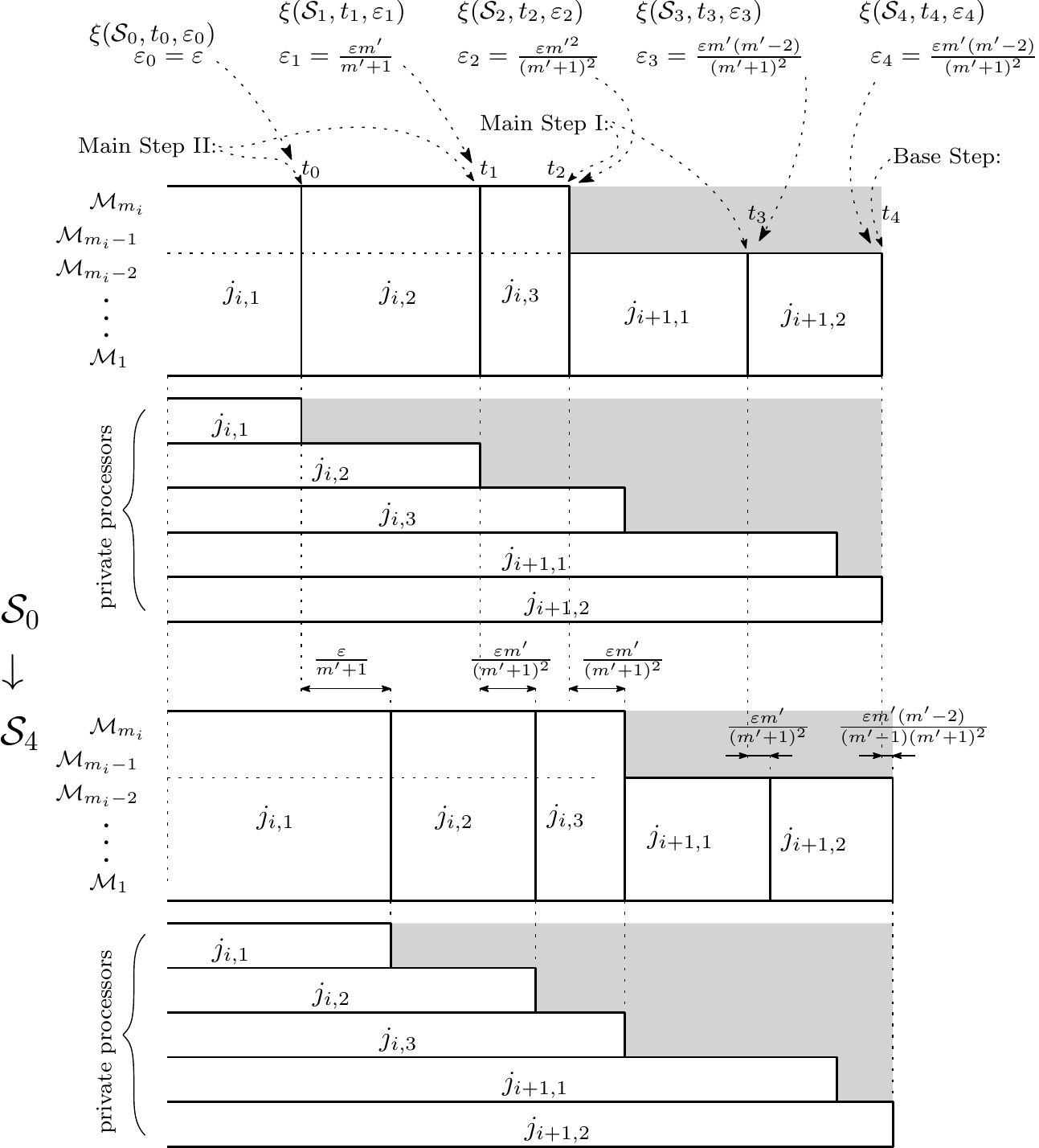}
\end{center}
\caption{In this example we consider two consecutive segments, which have three and two job pieces, respectively. We consider executing $\xi(\cS_i,t_i,\varepsilon_i)$ for $i=0,\ldots,4$, where $\varepsilon_0=\varepsilon$ is positive. All five modifications are doable but note that the job $j_{i,3}$ is synchronized in $\cS_4$ but is not synchronized in $\cS$. Hence, according to Condition~\eqref{eq:trans1:CaseI1} of Main Step~I that handles this modification, this is the maximum $\varepsilon>0$ for which all five modifications are doable.}
\label{fig:xi}
\end{figure}

Let $t_0<t_1<\cdots<t_q$, $q\geq 0$, be the last $q+1$ end points in the sequence $e_{1,1}<\cdots<e_{\ell,l(\ell)}$ of the schedule $\cS$. Let $I_1=(t_0,t_1),\ldots,I_q=(t_{q-1},t_q)$ be the last $q$ intervals of $\cS$. Let $j_i$, $m_i$, $m^+_i$, and $r_i$ be the job, the width, the factor, and the radius of the interval ending at $t_i$, $i=0,\ldots,q$. The $q+1$ step modification starts with $\cS_0=\cS$, $t=t_0$ and an $\varepsilon=\varepsilon_0$ such that
\begin{equation}\label{eps}
0<|\varepsilon|<\min_{i=0,\ldots,q}\{m_ir_i/2\},
\end{equation}
and recursively builds schedules $\cS_1,\ldots,\cS_{q+1}$ using the one step modifications just described such that $\cS_i=\xi(\cS_{i-1}, t_{i-1},\varepsilon_{i-1})$, where each subsequent value of $\varepsilon_i$ is computed on the basis of the previous one as follows: $\varepsilon_i=\nextEps{\varepsilon_{i-1}}$ for each $i\in\{1,\ldots,q+1\}$.
Finally, the $\cS_{q+1}=\xi(\cS_{q}, t_{q},\varepsilon_{q})$ is always the Base Step.
We say that the $q+1$ step modification is \emph{doable} if all its one step modifications are doable, i.e. all $\varepsilon_0, \ldots, \varepsilon_{q}$ satisfy appropriate condition in~\eqref{eq:trans1:base1}, \eqref{eq:trans1:CaseI1} and~\eqref{eq:trans1:CaseII2}. We later show that the initial choice of
$\varepsilon_0$ that meets (\ref{eps}) guarantees that the $q+1$ step modification is doable.
Observe that  by definition
\begin{equation}\label{epsi}
\varepsilon_i= \nextEps{\varepsilon_{i-1}} = \varepsilon_{i-1} \frac{m_{i}}{m^+_{i-1}}.
\end{equation}
Therefore the points $t_0, t_1, \ldots, t_q$ and the shifts $\varepsilon_0,\ldots, \varepsilon_{q}$ can be readily calculated from $\cS$ and $\varepsilon$.
We summarize this in the following corollary.
\begin{corollary} \label{cor:epsilon}
Consider a doable $q+1$ step modification that starts with $\cS_0=\cS$, $t=t_0$ and  $\varepsilon=\varepsilon_0$.
Then, 
\[\varepsilon_i = \varepsilon
\prod_{z=1}^{i} \frac{m_{z}}{m_{z-1}^+}
\]
for each $i\in\{1,\ldots,q\}$.
\qed
\end{corollary}
We remark that we will use later the fact that each $\varepsilon_i$ is linearly dependent on the $\varepsilon$.

\medskip
For a doable $q+1$ step modification that starts with $\cS_0=\cS$, $t=t_0$ and  $\varepsilon=\varepsilon_0$,
by~\eqref{eq:trans1:base3},~\eqref{eq:trans1:CaseI3} and~\eqref{eq:trans1:CaseII4}
the payoff can be written as follows
\begin{equation} \label{eq:payoff}
\Delta(\cS_{0},t_{0},\varepsilon_0)=\sum_{k=0}^{q}\payoff{\xi(\cS_{k},t_{k},\varepsilon_{k})} = \sum_{i=0}^{q} \frac{\varepsilon_i}{m_i^+} \left( \sum_{z=1}^{m_i}(\w{j_i}-\cm{z}) - \sum_{z=1}^{m_{i+1}}(\w{j_{i+1}}-\cm{z}) \right),
\end{equation}
where $m_{q+1}=0$.
We conclude from~\eqref{eq:payoff} the following.
\begin{corollary} \label{cor:payoff}
For a doable $q+1$ step modification that starts with $\cS_0=\cS$, $t=t_0$ and  $\varepsilon=\varepsilon_0$ it holds
\[\Delta(\cS_{0},t_{0},\varepsilon_0)=
\frac{\varepsilon_0}{m_0^+}\sum_{z=0}^{m_0}(\w{j_0}-\cm{z}) + \sum_{i=1}^{q} \left( \frac{\varepsilon_i}{m_i^+} - \frac{\varepsilon_{i-1}}{m_{i-1}^+} \right) \sum_{z=1}^{m_i}(\w{j_i}-\cm{z}).\]
\qed
\end{corollary}
Motivated by Corollaries~\ref{cor:epsilon} and~\ref{cor:payoff}, we introduce the following function for each $t\in T(\cS)$:
\[R(\cS,t) = \frac{1}{m_0^+}\sum_{z=1}^{m_0}(\w{j_0}-\cm{z}) + \sum_{i=1}^q \left(\frac{1}{m^+_i} \prod_{z=1}^{i}\frac{m_z}{m_{z-1}^+} - \frac{1}{m^+_{i-1}}\prod_{z=1}^{i-1}\frac{m_z}{m_{z-1}^+} \right) \sum_{z=1}^{m_i}(\w{j_i}-\cm{z}),\]
which we call the \emph{rate} of a doable $q+1$ step modification that starts with $\cS_0=\cS$, $t=t_0$ and  $\varepsilon=\varepsilon_0$.
We stress out that the rate is the same regardless of the value of $\varepsilon$ chosen for the modification.
In other words, the function $R$ depends only on the schedule $\cS$ and the time point $t\in T(\cS)$.
By Corollaries~\ref{cor:epsilon} and~\ref{cor:payoff} we obtain:
\begin{corollary} \label{cor:rate}
For a doable $q+1$ step modification that starts with $\cS_0=\cS$, $t=t_0$ and $\varepsilon=\varepsilon_0$,
$\Delta(\cS_{0},t_{0},\varepsilon_0)=\varepsilon\cdot R(\cS_0,t_0)$.
\qed
\end{corollary}

\medskip
Note that a doable $q+1$ step modification that starts with $\cS_0=\cS$, $t=t_0$ and  $\varepsilon=\varepsilon_0$ does not produce a feasible schedule $\cS'$. 
More precisely, the schedule $\cS'$ is not feasible since the total amount of the job $j_0$ equals $\p{j_0}+\varepsilon$ (note that this is the job $j_{i,1}$ in the example from Figure~\ref{fig:xi}) in $\cS'$.
We summarize the properties of $\cS'$ in the following lemmas.
\begin{lemma} \label{lem:xi-valid}
Let $\cS$ be a processor-descending and sequential schedule. A doable $q+1$ step modification that starts with $\cS_0=\cS$, $t=t_0$ and  $\varepsilon=\varepsilon_0$ that meets (\ref{eps}) results in $\cS'$ that satisfies the following conditions:
\begin{enumerate} [label={\normalfont{(\roman*)}}]
 \item\label{it:xi-v-0} the completion time of each job $j$ on each shared processor is smaller than or equal to $\complTime{\cS'}{j}{\Mpriv}$,
 \item\label{it:xi-v-1} the total execution time of each job $j\neq j_0$ equals $\p{j}$ in $\cS'$ ,
 \item\label{it:xi-v-2} the total execution time of $j_0$ in $\cS'$ is $\p{j_0}+\varepsilon$,
 \item\label{it:xi-v-4} no two pieces of jobs overlap in $\cS'$.
\end{enumerate}
\end{lemma}
\begin{proof} Assume $\varepsilon_0>0$ in the proof, the proof for $\varepsilon_0<0$ is similar and thus will be omitted.
The ends of the interval $I_i=(s_i,e_i)$ change to $(s'_i,e'_i)=I'_i$ as a result of the $q+1$ step modification as follows:
\begin{equation}\label{ends}
s'_i=s_i + \frac{\varepsilon_{i-1}}{m^+_{i-1}} \text{ and } e'_i=e_i + \frac{\varepsilon_{i}}{m^+_{i}}
\end{equation}
for $i=1,\ldots,q$. Thus
\begin{equation}
e'_i-s'_i=e_i-s_i+\frac{\varepsilon_{i-1}}{m^+_{i-1}}(\frac{m_i}{m^+_{i}}-1).
\end{equation}
For $m^+_i=m_i$, we have $e'_i-s'_i=e_i-s_i>0$. For $m^+_i=m_i+1$, we have
\begin{equation}\label{E}
e'_i-s'_i=e_i-s_i-\frac{\varepsilon_{i-1}}{m^+_{i-1}m^+_{i}}.
\end{equation}
Since $\varepsilon>\frac{\varepsilon_{i-1}}{m^+_{i-1}}$, and by (\ref{eps}) $m^+_i(e_i-s_i)>\varepsilon$, we have $e'_i-s'_i>0$ for $\varepsilon>0$. Thus, by the one step modifications, \ref{it:xi-v-4} holds.

By (\ref{ends}) the execution of $j_i$ is reduced (this does not happen for $j_0$ for which the reduction is 0) by
\begin{equation}
\varepsilon_{i-1}\frac{m_i}{m^+_{i-1}}=\varepsilon_i,
\end{equation}
and it increases by
\begin{equation}
\varepsilon_i\frac{m_i}{m^+_{i}}
\end{equation}
on shared processors.
For $m^+_i=m_i$ the two are equal, and for $m^+_i=m_i+1$, the private processor $\Mpriv_{j_0}$ of job $j_0$ gets $\varepsilon_{i}\frac{1}{m^+_{i}}$ of that job. Thus \ref{it:xi-v-1} and \ref{it:xi-v-2} hold.

In Base Step, the job $j_{\ell}$ is synchronized due to~\eqref{eq:trans1:base2}.
Similarly, in Main Step~II, also $j_i$ completes both on shared processor and on its private processor at the same time according to~\eqref{eq:trans1:CaseII3}.
In Main Case~I, the completion time of $j$ is set in~\eqref{eq:trans1:CaseI2} and this does not exceed $\complTime{\cS'}{j}{\Mpriv}$ by definition of $r_i$ in the right hand side inequality in~\eqref{eq:trans1:CaseI1}.
For all remaining jobs their completion times on all processors remain unchanged, which proves~\ref{it:xi-v-0}.
\end{proof}

The second lemma shows a sufficient condition for $\varepsilon$ to make $q+1$ step modification that starts with $\cS_0=\cS$, $t=t_0$ and  $\varepsilon_0=\varepsilon$
doable.
Recall that $\cS'$ produced by the $q+1$ step modification is not feasible however, by Lemma~\ref{lem:xi-valid}, the only reason for that is that the total execution time of the job $j_0$ is incorrect in $\cS'$, i.e., it equals $\p{j_0}+\varepsilon$ instead of $\p{j_0}$.
For this reason we introduce notation $\cS'_{-j}$, for a job $j$, to denote a schedule obtained from $\cS'$ by removing all pieces of $j$ from shared processors and by removing the private processor of $j$.
Note that $\cS'_{-j}$ is then a feasible schedule for the instance $\jobs\setminus\{j\}$.
Hence, the second lemma also shows the difference between the total weighted overlap $\tct{\cS'_{-j}}$ of $\cS'_{-j}$, which gives the sum of total overlaps of all jobs in $\cS'$ except of $j$,
and the total weighted overlap $\tct{\cS}$ of $\cS$. 
\begin{lemma} \label{lem:xi-epsilon}
Let $\cS$ be a processor-descending sequential schedule. Let $\varepsilon$ meet (\ref{eps}).
Then, both $q+1$ step modification that starts with $\cS_0=\cS$, $t=t_0$ and  $\varepsilon_0=-\varepsilon$  and $q+1$ step modification that starts with $\cS_0=\cS$, $t=t_0$ and or  $\varepsilon_0=\varepsilon$ are doable, and we have
  \begin{equation} \label{eq:xi-e-4}
    \tct{\cS'_{-j_0}}=\tct{\cS}+ \Delta(\cS_{0},t_{0},\varepsilon_0)
    -\sum_{i=1}^m\overlap{\cS}{j_0,\Mshared_i}(\w{j_0}-\cm{i}) - \frac{\varepsilon}{m_{0}^+}\sum_{z=1}^{m_0}(\w{j_0}-\cm{z}),
  \end{equation}
for the resulting schedule $\cS'$.
\end{lemma}
\begin{proof}
Assume $\varepsilon_0>0$ in the proof, the proof for $\varepsilon_0<0$ is similar and thus will be omitted.
We first prove that
\begin{equation} \label{E1}
m^+_i(e'_i-s'_i)>\varepsilon_i
\end{equation} 
for $i\in\{1,\ldots,q\}$. This holds for $m^+_i=m_i$ since then $e_i-s_i=e'_i-s'_i$ and by $\varepsilon\geq \varepsilon_i$. Suppose $m^+_i=m_i+1$. By (\ref{E}) and (\ref{E1}) we need to show
\begin{equation}
m^+_i(e_i-s_i)> \varepsilon_i+\frac{\varepsilon_{i-1}}{m^+_{i-1}m^+_i}.
\end{equation}
To that end we observe that
\begin{equation}
\frac{\varepsilon_{i-1}m^+_i}{m^+_{i-1}}>\varepsilon_i+\frac{\varepsilon_{i-1}}{m^+_{i-1}m^+_i}=\frac{\varepsilon_{i-1}m_i}{m^+_{i-1}}+\frac{\varepsilon_{i-1}}{m^+_{i-1}m^+_i}.
\end{equation}
Thus it suffices to show that
\begin{equation}
m^+_i(e_i-s_i)> \frac{\varepsilon_{i-1}m^+_i}{m^+_{i-1}},
\end{equation}
or equivalently
\begin{equation}
e_i-s_i> \frac{\varepsilon_{i-1}}{m^+_{i-1}}.
\end{equation}
By multiplying both sides of the last inequality by $m_i$ we get
\begin{equation}
m_i(e_i-s_i)> \frac{\varepsilon_{i-1}m_i}{m^+_{i-1}}=\varepsilon_i.
\end{equation}
This last inequality holds since $\varepsilon\geq \varepsilon_i$ and (\ref{eps}) holds for $\varepsilon$. We also prove, a similar proof for $m^+_i\complTime{\cS}{j}{\Mpriv}-e'_{i}>\varepsilon_i$ will be omitted, that
\begin{equation} \label{E10}
m^+_i(p_j-e'_i)>\varepsilon_i.
\end{equation} 
Since $m^+_i\geq m_i$, it suffices to show that
\begin{equation}
m_i(p_j-e_i)>\varepsilon_i+\varepsilon_i\frac{m_i}{m^+_i}.
\end{equation}
The last inequality holds since $m_i(p_j-e_i)\geq 2(m_i r_i/2)>2\varepsilon\geq \varepsilon_i$ by (\ref{eps}). This completes the proof of the first part of the lemma. Observe that $\varepsilon_i$ does not reach neither end of doable intervals.

For the proof of the second part, let for brevity $\jobs'=\{j_0,j_1,\ldots,j_q\}$.
We have
\begin{align} \label{eq:payoff-barj1}
\begin{split}
\tct{\cS'_{-j_0}} & = \sum_{j\in\jobs\setminus\{j_0\}}\sum_{i=1}^m \overlap{\cS'_{-j_0}}{j,\Mshared_i}(\w{j}-\cm{i}) \\
            & = \sum_{j\in\jobs\setminus\jobs'}\sum_{i=1}^m \overlap{\cS}{j,\Mshared_i}(\w{j}-\cm{i}) + \sum_{j\in\jobs'\setminus\{j_0\}}\sum_{i=1}^m \overlap{\cS'_{-j_0}}{j,\Mshared_i}(\w{j}-\cm{i}).
\end{split}
\end{align}
Consider any $j_i\in\jobs'\setminus\{j_0\}$.
The total weighted overlap of $j_i$ is the same in $\cS_{-j_0}$ as in $\cS$ except for the shift in its piece performed by the $i$-th and $(i-1)$-st modifications (see also Corollary~\ref{cor:payoff}):
\begin{align} \label{eq:payoff-barj3}
\begin{split}
\sum_{z=1}^m \overlap{\cS_{-j_0}}{j_i,\Mshared_z} (\w{j_i}-\cm{z}) = & \sum_{z=1}^m \overlap{\cS}{j_i,\Mshared_z} (\w{j_i}-\cm{z}) \\
              & + \left(\frac{\varepsilon_i}{m_i^+} - \frac{\varepsilon_{i-1}}{m_{i-1}^+}\right)\sum_{z=1}^{m_i} (\w{j_i}-\cm{z}).
\end{split}
\end{align}
By~\eqref{eq:payoff-barj1} and~\eqref{eq:payoff-barj3} applied to all jobs in $\jobs'\setminus\{j_0\}=\{j_1,\ldots,j_q\}$ we obtain
\begin{align} \label{eq:payoff-barj2}
\begin{split}
\tct{\cS'_{-j_0}} = \tct{\cS} & - \sum_{i=1}^m \overlap{\cS}{j_0,\Mshared_i}(\w{j_0}-\cm{i}) \\
                                  & + \sum_{i=1}^{q}\left(\frac{\varepsilon_i}{m_i^+} - \frac{\varepsilon_{i-1}}{m_{i-1}^+}\right)\sum_{z=1}^{m_i} (\w{j_i}-\cm{z}).
\end{split}
\end{align}
By~Corollary~\ref{cor:payoff},~\eqref{eq:payoff-barj2} and $\varepsilon_0=\varepsilon$,
\begin{align} \label{eq:payoff-barj5}
\begin{split}
\tct{\cS'_{-j_0}} = \tct{\cS} & - \sum_{i=1}^m \overlap{\cS}{j_0,\Mshared_i}(\w{j_0}-\cm{i}) \\
                                  & + \Delta(\cS_{0},t_{0},\varepsilon_0)
                                  - \frac{\varepsilon}{m_{0}^+}\sum_{z=1}^{m_0}(\w{j_0}-\cm{z}).
\end{split}
\end{align}
which proves~\eqref{eq:xi-e-4} and completes the proof of the lemma.
\end{proof}

\subsection{Splits} \label{sec:split}

Suppose that $\cS$ is a  processor-descending and sequential schedule.
We say that a job $j$ has a $(I,I')$-\emph{split} if $I$ and $I'$ are two pieces of $j$ executing in two different segments.
We assume that $I'$ is to the right of $I$.
Given that $\cS$ has such a job $j$ with a $(I,I')$-split, we introduce the following schedule transformation that we call a $(I,I',\varepsilon)$-\emph{transfer}.
Although this modification works for an arbitrary split, we will be particularly interested in our analysis in the case when the $(I,I')$-split is the rightmost.
Let $q\geq 0$ be the number of intervals (job pieces) to the right of $I'$ in $\cS$. Consider $\varepsilon$ such that
\begin{equation} \label{epsA}
0<|\varepsilon|< \min\left\{|I|,  \frac{|I'|}{1+1/m_0^+}, \min_{i=0,\ldots,q}\{m_ir_i/2\}\right\}
\end{equation}
where $m^+_0$ is the factor of interval $I'=(s_0,e_0)$. Let $m_0$ be the width of $I'$.
The modification is composed of the following steps.
\begin{enumerate} [label={\normalfont{(T\arabic*)}}]
 \item\label{it:transfer1} Obtain a schedule $\cS'$ by performing $q+1$ step modification with $\cS$, $t=e_0$, and $\varepsilon$.
 \item\label{it:transfer2a} If $\varepsilon>0$, then change in $\cS'$ the completion time of the piece in $I$ of the job $j$ from $y$ to $y-\varepsilon$ on the processor $\Mshared_{m_0+1}$, where $y$ is the right endpoint of $I$.
 \item\label{it:transfer2b} If $\varepsilon<0$, then add a piece of the job $j$ of length $\card{\varepsilon}$ to the shared processor $\Mshared_{m_0+1}$ in time interval $(s_0,s_0+\card{\varepsilon})$. %where $l'$ is the left endpoint of $I'$.
 \item\label{it:transfer3} Call $\procSequential(\cS')$ to make each segment of the new schedule sequential, and return $\cS'$.
\end{enumerate}

The $(I,I')$-transfer is illustrated in Figure~\ref{fig:split}.
\begin{figure}[ht!]
\begin{center}
 \includegraphics[scale=1.1]{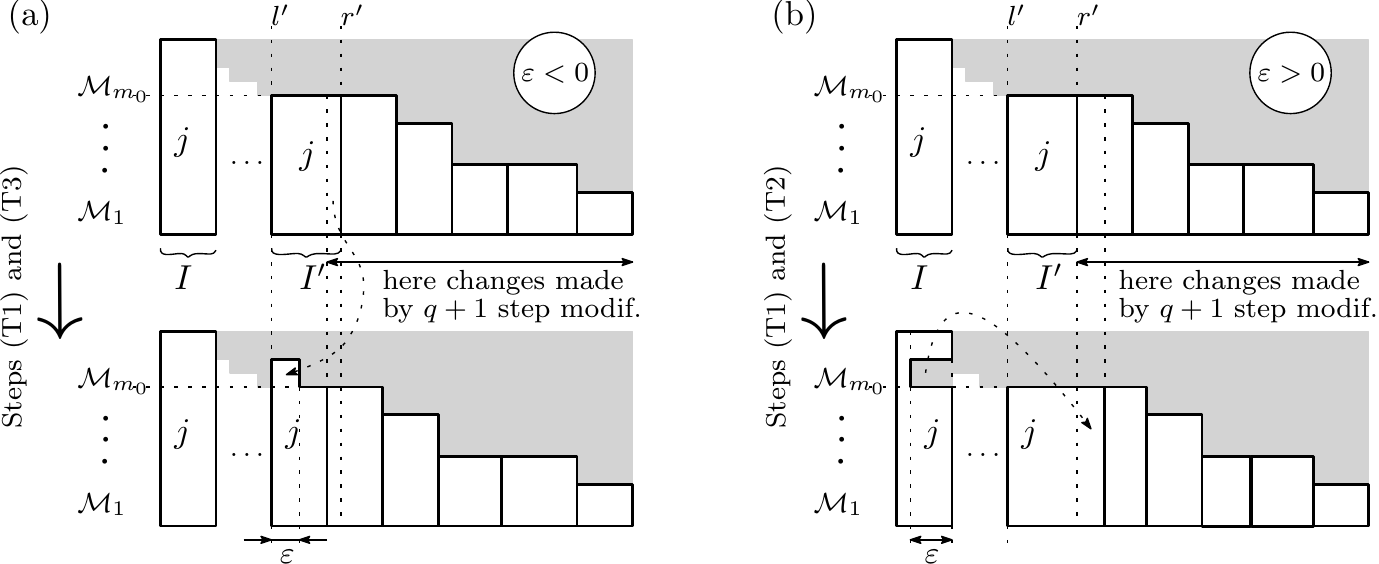}
\end{center}
\caption{A $(I,I',\varepsilon)$-transfer:
        (a) if $\varepsilon<0$, then Steps~\ref{it:transfer1} and~\ref{it:transfer2b} modify the schedule shown on top to the one on bottom;
        % --- note that for $\varepsilon>\varepsilon^-(\cS,r)$ a new rightmost split with gap equal to one is introduced for the job $j$;
        (b) if $\varepsilon>0$, then Steps~\ref{it:transfer1} and~\ref{it:transfer2a} are applied.}
\label{fig:split}
\end{figure}

\begin{lemma} \label{lem:transfer}
If $\cS$ is an optimal processor-descending and sequential schedule with a job $j$ having the right-most $(I,I')$-split, then both
the $q+1$ step modification that starts with $\cS_0=\cS$, $t_0=e_0$, $\varepsilon_0=\varepsilon$, and the $q+1$ step modification that starts with $\cS_0=\cS$, $t_0=e_0$, $\varepsilon_0=-\varepsilon$, where $\varepsilon$ satisfies condition (\ref{epsA}), are doable and produce a processor-descending and sequential schedule $\cS'$. Moreover, for $\varepsilon<0$ we have $\cS'$ shorter than $\cS$.
\end{lemma}
\begin{proof}
In Step~\ref{it:transfer2a}, $\varepsilon$ must not exceed $\card{I}$ as this is the length of the piece of $j$ executing in the interval $I$.
By~\eqref{eq:trans1:base2},~\eqref{eq:trans1:CaseI2} and~\eqref{eq:trans1:CaseII3} in Step~\ref{it:transfer2b}, we need $\card{\varepsilon}\leq |I'|-\card{\varepsilon}/m_1^+$.
Thus, we obtain a condition
\[\card{\varepsilon}\leq \frac{|I'|}{1+1/m_0^+}.\]

By Lemma \ref{lem:xi-epsilon} both $q+1$ step modifications are doable since (\ref{epsA}) implies (\ref{eps}). Moreover, for $\varepsilon<0$ we have $\cS'$ shorter that $\cS$ due to Base Step of the modification.
\end{proof}

\begin{lemma} \label{lem:rate}
Suppose that $\cS$ is a processor-descending and sequential schedule with a job $j$ having the right-most $(I,I')$-split and consider any $\varepsilon$ that meets (\ref{epsA}).
Then, the processor-descending and sequential schedule $\cS'$ resulting from the $(I,I',\varepsilon)$-transfer satisfies
\[\tct{\cS'} = \tct{\cS} + \varepsilon\cdot (R(\cS,t)-\w{j}+\cm{m_0+1}),\]
where $m_0$ is the width of the interval $I'=(s',e'=t)$ in $\cS$. Moreover, for an optimal $\cS$, 
\[R(\cS,t)-\w{j}+\cm{m_0+1}=0.\]
\end{lemma}
\begin{proof}
Consider first an arbitrary $\varepsilon$ that meets (\ref{epsA}). By Lemma~\ref{lem:transfer}, the $(I,I',\varepsilon)$-transfer is doable.
Let $j$ be the job with $(I,I')$-split modified in Step~\ref{it:transfer2a} of Procedure~$\procTransfer$.
We first calculate the sum of $\overlap{\cS'}{j,\Mshared_i}$ taken over all shared processors $\Mshared_i$.
This value is similar to that in $\cS$, except for the two changes introduced to $j$ in Steps~\ref{it:transfer1}, \ref{it:transfer2a} and~\ref{it:transfer2b} of Procedure~$\procTransfer$.
In Steps~\ref{it:transfer2a} and~\ref{it:transfer2b}, the total weighted overlap of $j$ changes by
\[-\varepsilon(\w{j}-\cm{m_0+1}).\]
In Step~\ref{it:transfer1}, it changes by
\[\frac{\varepsilon}{m_0^+} \sum_{z=1}^{m_0}(\w{j}-\cm{z}),\]
where $m_0$ and $m_0^+$ are the width and the factor of $I'$. 
%The first modification of $\xi(\cS,r,\varepsilon)$ performed in Step~\ref{it:transfer1} of Procedure~$\procTransfer$.
Indeed, this follows from the fact that the transformation changes only the right endpoint of the piece of $j$ executing in $I'=(s',e')$, and due to~\eqref{eq:trans1:base2},~\eqref{eq:trans1:CaseI2} and~\eqref{eq:trans1:CaseII3} this value changes, on each machine $\Mshared_1,\ldots,\Mshared_{m_1}$, by $\varepsilon/m_1^+$ since this is done in the first step of the $q+1$ step modification that starts with $\cS_0=\cS$, $t_0=e'$, and $\varepsilon_0=\varepsilon$.
Thus we obtain
\begin{align} \label{eq:transfer-j}
\begin{split}
\sum_{z=1}^{m}\overlap{\cS'}{j,\Mshared_z}(\w{j}-\cm{z}) = & \sum_{z=1}^{m}\overlap{\cS}{j,\Mshared_z}(\w{j}-\cm{z}) \\ 
                                                       & - \varepsilon(\w{j}-\cm{m_0+1}) + \frac{\varepsilon}{m_0^+} \sum_{z=1}^{m_0}(\w{j}-\cm{z}).
\end{split}
\end{align}
The total weighted overlap of $\cS'$ can be expressed as
\[\tct{\cS'} = \tct{\cS'_{-j}} + \sum_{z=1}^{m}\overlap{\cS'}{j,\Mshared_z}(\w{j}-\cm{z}).\]
By Lemma~\ref{lem:xi-epsilon} (where $j_0$ is taken to be $j$) and~\eqref{eq:transfer-j},

\[\tct{\cS'} = \tct{\cS}+
\Delta(\cS,t_0,\varepsilon) - \varepsilon(\w{j}-\cm{m'+1}).\]
By Corollary~\ref{cor:rate},
\[\tct{\cS'} = \tct{\cS} + \varepsilon\cdot (R(\cS,t_0)-\w{j}+\cm{m'+1}).\]

Note that the value of the expression $R(\cS,t_0)-\w{j}+\cm{m_0+1}$ depends only on the schedule $\cS$ and the point $e'$, where $I'=(s',e'=t_0)$.
If this value is negative, then by Lemma~\ref{lem:transfer}, $(I,I',\varepsilon<0)$-transfer is doable and results in a feasible schedule $\cS'$, which satisfies by Lemma~\ref{lem:rate}: $\tct{\cS'}>\tct{\cS}$. Thus, a contradiction.
If this value is positive, then again by Lemma~\ref{lem:transfer}, $(I,I',\varepsilon>0)$-transfer is doable and results in a feasible schedule $\cS'$, which satisfies by Lemma~\ref{lem:rate} again the desired inequality: $\tct{\cS'}>\tct{\cS}$. Thus, again a contradiction.
If, however this value equals $0$, then we can arbitrarily perform either $(I,I',\varepsilon>0)$-transfer or $(I,I',\varepsilon<0)$-transfer and Lemmas~\ref{lem:rate} and~\ref{lem:transfer} guarantee that we obtain some schedule $\cS'$ with $\tct{\cS'}=\tct{\cS}$. This proves the lemma.
\end{proof}

\begin{lemma} \label{lem:no-splits}
There exists an optimal schedule that is processor-descending, sequential and has no job splits.
\end{lemma}
\begin{proof}
Consider an optimal schedule $\cS$ that is processor-descending and sequential. Without loss of generality we may assume that $\cS$ has the minimum makespan among all optimal processor-descending and sequential schedules. Suppose for a contradiction that $(I,I')$ is the rightmost split in $\cS$. By Lemma~\ref{lem:rate}, the processor-descending and sequential schedule $\cS'$ resulting from the $(I,I',\varepsilon<0)$-transfer satisfies $\tct{\cS'} = \tct{\cS}$. However, since $\varepsilon<0$, $\cS'$ is shorter than $\cS$ which contradicts our choice of $\cS$.
\end{proof}

\subsection{Synchronization} \label{sec:synchronization}

Consider a processor-descending and sequential schedule $\cS$ that has no splits and let $j$ be the \emph{last} job in $\cS$ that is not synchronized, i.e., the job that has the greatest completion time on shared processors among jobs that are not synchronized.
Suppose that $j$ is present on $m_0\geq 1$ shared processors.
Since $\cS$ is sequential, $j$ starts and ends on $\Mshared_1,\ldots,\Mshared_{m_0}$ at time points $s$ and $e$, respectively. Let $q$ be the number of intervals to the right of the interval $I=(s,e)$ in which the piece of $j$ executes.
Define
\begin{equation} \label{epsB}
0<|\varepsilon|< \min\left\{ m_0(e-s), \frac{m_0}{m_0+1}\left(\complTime{\cS}{j}{\Mpriv}-e\right), \min_{i=0,\ldots,q}\{m_ir_i/2\} \right\}
\end{equation}
The following operation that we call a \emph{$j$-synchronization}, performs a transition from $\cS$ to a schedule $\cS'$.
\begin{enumerate} [label={\normalfont{(S\arabic*)}}]
 \item\label{synchronization:eps} 
 If $R(\cS,t=e)>0$, then let $\varepsilon>0$ and otherwise let $\varepsilon<0$, where $\varepsilon$ satisfies (\ref{epsB}).
 \item\label{synchronization:shift} Perform $q+1$ step modification that starts with $\cS$, $e$, and $\varepsilon$.
 \item\label{synchronization:j} Obtain $\cS'$ by setting the completion time of $j$ on the private processor to $\complTime{\cS'}{j}{\Mpriv}:=\complTime{\cS}{j}{\Mpriv}-\varepsilon$.
\end{enumerate}

\begin{lemma} \label{Rzero}
Suppose $\cS$ is an optimal processor-descending and sequential schedule with no splits, and with job $j$ which is not synchronized and done in $I=(s,e)$ on shared processors. Then
$R(\cS,t=e)=0$
\end{lemma}

\begin{proof}
By Lemma \ref{lem:xi-epsilon}, the $q+1$ step modification called by the $j$-synchronization is doable since (\ref{epsB}) implies (\ref{eps}). Also no more than $m_0(e-s)>|\varepsilon|$ of $j$ can be moved from the $m_0$ shared processors in the interval $I=(s,e)$ to the job's private processor, and no more than $\frac{m_0}{m_0+1}\left(\complTime{\cS}{j}{\Mpriv}-e\right)>|\varepsilon|$ can be moved from the job's private processor to the $m_o$ shared processors. Thus the choice of $\varepsilon$ guarantees that $\cS'$ is feasible.
%It remain to analyze the total weighted overlap of $\cS'$ with respect to $\cS$.
For each shared processor $i\in\{1,\ldots,m_0\}$, the execution time of $j$ on $\Mshared_i$ changes by $\varepsilon/m_0$ (if $\varepsilon<0$, then the execution time decreases, otherwise it increases).
Hence, for each such $i$,
$\overlap{\cS'}{j,\Mshared_i}=\overlap{\cS}{j,\Mshared_i}+\varepsilon/m_0$.
We can hence represent the total weighted overlap of $\cS'$ as follows:
\begin{align*}
\tct{\cS'}  & = \tct{\cS'_{-j}}+\sum_{i=1}^{m_0}\overlap{\cS'}{j,\Mshared_i}\cdot(\w{j}-\cm{i}) \\
            & = \tct{\cS'_{-j}}+\sum_{i=1}^{m_0}\overlap{\cS}{j,\Mshared_i}\cdot(\w{j}-\cm{i})+\frac{\varepsilon}{m_0}\sum_{i=1}^{m_0}(\w{j}-\cm{i}).
\end{align*}
By Lemma~\ref{lem:xi-epsilon},
\[\tct{\cS'} = \tct{\cS}+\Delta(\cS,t,\varepsilon).\]
Note that in the above $\overlap{\cS}{j,\Mshared_i}=0$ for each $i>m_0$ and hence $\sum_{i=1}^m\overlap{\cS}{j,\Mshared_i}(\w{j}-\cm{i})=\sum_{i=1}^{m_0}\overlap{\cS}{j,\Mshared_i}\cdot(\w{j}-\cm{i})$.
Thus, by Corollary~\ref{cor:rate}, $\tct{\cS'}=\tct{\cS}+\varepsilon R(\cS,t)$. By definition of $j$-synchronization we have $R(\cS,t)=0$ since otherwise $\tct{\cS'}>\tct{\cS}$ which contradicts
our choice of $\cS$. This proves the lemma.
\end{proof}

We are now ready to complete the proof that there exist optimal schedules that are synchronized.
\begin{proof}[Proof of Theorem~\ref{thm:synchronized}]
Consider an optimal schedule $\cS$ that is processor-descending, sequential and without splits. Without loss of generality we may assume that $\cS$ has minimum makespan among all
optimal processor-descending, sequential schedules and without splits. Suppose for a contradiction that $\cS$ is not synchronized. Let $j$ be the last job that is not synchronized, and let a piece of $j$ be executed in the interval $I=(s,e)$ on shared processors in $\cS$. By Lemma \ref{Rzero}, $R(\cS,t=e)=0$. Do the $j$-synchronization with $\varepsilon<0$ and meeting the condition~\eqref{epsB}. For the resulting schedule we have $\tct{\cS'} = \tct{\cS}$ according to Corollary~\ref{cor:rate}. Moreover, $\cS'$ is  processor-descending, sequential schedule and without splits.  However, since $\varepsilon<0$, $\cS'$ is shorter than $\cS$ which contradicts our choice of $\cS$.
\end{proof}

\section{Conclusions and Open Problems} \label{sec:conclusions}

Our first open problem regards the complexity of the problem.
The complexity question remains open even for the single machine case, i.e., the $m=1$ case~\cite{DK17}.
Note however that the problem with \emph{SP} jobs mode (recall that this is the problem variant where each job may use at most one shared processor) is NP-complete in the strong sense~\cite{DK16}, and no approximation algorithm with guaranteed worst case ratio is know for the problem.
The structural characterization shown in this paper for the \emph{MP} job mode (recall that this is the problem variant where each job may use many, possibly all, shared processor simultaneously) indicates, intuitively speaking, that in this mode schedules for $m>1$ shared processors `resemble' schedules on a single shared processor in the sense that in both cases the jobs that appear on the shared processors have certain ordering: once one job finishes on all shared processors it uses, another job starts exclusively using all shared processors it requires. However the numbers of shared processors used by the jobs may be different since the jobs later in the sequence may consider some shared processor too expensive to use.
Therefore, with respect to that the \emph{SP} and \emph{MP} modes behave very differently.

Our approximation ratio of $\frac{1}{2}+\frac{1}{4(m+1)}$ obtained for arbitrary number $m\geq 1$ of  shared processors improves the previously known approximation ratio, see~\cite{DK17}, from $\frac{1}{2}$ to $\frac{5}{8}$ in the single shared processor case.
We leave an open question whether the approximation ratio provided by Theorem~\ref{thm:approximation} is the best possible, both for multiple shared processors and for a single shared processor.

\section*{Acknowledgements}
This research has been supported by the Natural Sciences and Engineering Research
Council of Canada (NSERC) Grant OPG0105675.

\bibliographystyle{plain}
\bibliography{references}
\end{document}